\documentclass[12pt]{article}
\usepackage{amssymb}
\usepackage{amsfonts}
\usepackage{amsmath}
\usepackage{graphicx}
\usepackage{geometry}
\usepackage{bussproofs}
\usepackage{verbatim}
\usepackage[color,all]{xy}
\usepackage{color}
\usepackage{xspace}
\usepackage{float}
\usepackage{euscript}
\usepackage{hyperref}

\newtheorem{theorem}{Theorem}

\newtheorem{definition}[theorem]{Definition}
\newtheorem{example}[theorem]{Example}

\newtheorem{lemma}[theorem]{Lemma}
\newtheorem{notation}[theorem]{Notation}

\newtheorem{proposition}[theorem]{Proposition}
\newtheorem{remark}[theorem]{Remark}

\newenvironment{proof}[1][Proof]{\textbf{#1.} }{\ \rule{0.5em}{0.5em}}
\input{tcilatex}

\newcommand \ucomment[1]{}

\newcommand \citep[1]{\cite{#1}}

\newcommand \LongVersion[1]{}

\newcommand \tif {\text{ if }}
\newcommand \tand {\textrm{ and }}
\newcommand \tor {\text{ or }}

\newcommand \fail {\mathbb{F}}
\newcommand \terms {\mathcal{T}(\Sigma,\mathcal{X})}

\begin{document}

\title{Computer-Aided Derivation of Multi-scale Models: A Rewriting Framework \thanks{%
This work is partially supported by the European Territorial Cooperation
Programme INTERREG IV A France-Switzerland 2007-2013.}
}
\author{ Bin Yang \\
{\small Department of Applied Mathematics}\\
[-0.8ex] {\small Northwestern Polytechnical University }\\
[-0.8ex] {\small 710129 Xi'an Shaanxi, China }\\
{\small University of Franche-Comt\'{e},}\\
[-0.8ex] {\small 26 chemin de l'Epitaphe, 25030 Besan\c{c}on Cedex, France}\\
[-0.8ex] {\small \texttt{bin.yang@femto-st.fr}}\\
\and Walid Belkhir \\
{\small INRIA Nancy - Grand Est, }\\
[-0.8ex]{\small  CASSIS project, }\\
[-0.8ex] {\small 54600 Villers-l\`es-Nancy, France }\\
{\small \texttt{walid.belkhir@inria.fr}} 
\and Michel Lenczner \\
{\small FEMTO-ST, D\'{e}partement Temps-Fr\'{e}quence,}\\
[-0.8ex] {\small University of Franche-Comt\'{e},}\\
[-0.8ex] {\small 26 chemin de l'Epitaphe, 25030 Besan\c{c}on Cedex, France}\\
[-0.8ex] {\small \texttt{michel.lenczner@utbm.fr} }}
\maketitle


\abstract{
We introduce a framework for computer-aided derivation of multi-scale
models. It relies on a combination of an asymptotic method used in the field
of partial differential equations with term rewriting techniques coming from
computer science.
  In our approach, a multi-scale model derivation is
characterized by the features taken into account in the asymptotic analysis.
Its formulation consists in a derivation of a reference model\textit{\ }associated to an elementary nominal model, and in a set of transformations
to apply to this proof until it takes into account the wanted features. In
addition to the reference model proof, the framework includes first order
rewriting principles designed for asymptotic model derivations, and second
order rewriting principles dedicated to transformations of model
derivations. We apply the method to generate a family of homogenized models
for second order elliptic equations with periodic coefficients that could be
posed in multi-dimensional domains, with possibly multi-domains and/or thin
domains.

}

\section{Introduction}

There is a vast literature on multi-scale methods for partial differential
equations both in applied mathematics and in many modeling areas. Among all
developed methods, asymptotic methods occupy a special place because they
have rigorous mathematical foundations and can lead to error estimates based
on the small parameters involved in the approach. This is a valuable aspect
from the model reliability point of view. They have been applied when a
physical problem depends on one or more small parameters which can be some
coefficients or can be related to the geometry. Their principle is to
identify the asymptotic model obtained when the parameters tend to zero. For
instance, this method applies in periodic homogenization, i.e. to systems
consisting of a large number of periodic cells, the small parameter being
the ratio of the cell size over the size of the complete system, see for
instance \cite{BenLio, CioDon, JikKoz}. Another well-developed case is when
parts of a system are thin, e.g. thin plates as in [Cia], that is to say
that some of their dimensions are small compared to others. A third kind of
use is that of strongly heterogeneous systems e.g. \cite{BouBel}, i.e. when
equation coefficients are much smaller in some parts of a system than in
others. These three cases can be combined in many ways leading to a broad
variety of configurations and models. In addition, it is possible to take
into account several nested scales and the asymptotic characteristics can be
different at each scale: thin structures to a scale, periodic structures to
another, etc.... It is also possible to cover cases where the asymptotic
phenomena happen only in certain regions or even are localized to the
boundary. Moreover, different physical phenomena can be taken into account:
heat transfer, solid deformations, fluid flow, fluid-structure interaction
or electromagnetics. In each model, the coefficients can be random or
deterministic. Finally, different operating regimes can be considered as the
static or the dynamic regimes, or the associated spectral problems. Today,
there exists a vast literature covering an impressive variety of
configurations.

\bigskip

Asymptotic methods, considered as model reduction techniques, are very
useful for complex system simulation and are of great interest in the software design
 community. They enjoy a number of advantages. The resulting models
are generally much faster (often by several order of magnitude -- depending
on the kind of model simplification --) to simulate than the original one
and are fully parameterized.\ In addition, they do not require any long
numerical calculation for building them, so they can be inserted into
identification and optimization loops of a design process.\ Finally, they
are of general use and they can be rigorously applied whenever a model
depends on one or several small parameters and the error between their
solution and nominal model solution can be estimated.

\bigskip

Despite these advantages, we observe that the asymptotic modeling techniques
have almost not been transferred in general industrial simulation software
 while numerical techniques, as for instance the Finite
Element Method, have been perfectly integrated in many design tools. The
main limitation factor for their dissemination is that each new problem
requires new long hand-made calculations that may be based on a relatively
large variety of techniques. In the literature, each published paper focus
on a special case regarding geometry or physics, and no work is oriented in
an effort to deal with a more general picture. Moreover, even if a large
number of models combining various features have already been derived, the
set of already addressed problems represents only a tiny fraction of those
that could be derived from all possible feature combinations using existing
techniques.

\bigskip

Coming to this conclusion, we believe that what prevents the use of
asymptotic methods by non-specialists can be formulated as a scientific
problem that deserves to be posed. It is precisely the issue that we discuss
in this paper. We would like to establish a mathematical framework for
combining asymptotic methods of different nature and thus for producing a
wide variety of models. This would allow the derivation of complex
asymptotic models are made by computers. In this paper, we present first
elements of a solution by combining some principles of asymptotic model
derivations and \textit{rewriting} methods issued from computer science.

\bigskip

In computer science equational reasoning is usually described by rewrite
rules, see \cite{BaaNip} for a classical reference. A rewrite rule $%
t\rightarrow u$ states that every occurrence of an instance of $t$ in a term
can be replaced with the corresponding instance of $u$. Doing so, a proof
based on a sequence of equality transformations is reduced to a series of
rewrite rule applications. Rules can have further conditions and can be
combined by specifying strategies which specify where and when to apply
them, see for instance \cite{Terese03, rhoCalIGLP-I+II-2001, Cirstea200551,
BKK+99, RewriteStrat_CHK2003}

\bigskip

The method developed in this paper is led by the idea of derivating models
by generalization. For this purpose, it introduces a \textit{reference model}
with its derivation and a way to generate generalizations to cover more
cases. The level of detail in the representation of mathematical objects
should be carefully chosen. On the one hand it should have enough precision
to cover fairly wide range of models and on the other hand calculations
should be reasonably sized. The way the generalizations are made is important
so that they could be formulated in a single framework.

\bigskip

In this paper, we select as \textit{reference problem} that of the
periodic homogenization of a scalar second order elliptic equation posed in
a one-dimension domain and with Dirichlet boundary conditions. Its
derivation is based on the use of the two-scale transform operator
introduced in \cite{ArbDou}, and reused in \cite{BouLuc}. We quote that
homogenization of various problems using this transformation was performed
according to different techniques in \cite{Len97, Len06, LenSmi07, Cas00,
CioDam02, CioDam08}. Here, we follow that  of \cite{LenSmi07}, so a number of
basic \textit{properties} coming from this paper are stated and considered
as the building blocks of the proofs. The complete derivation of the model
is organized into seven \textit{lemmas} and whose proof is performed by a
sequence of applications of these \textit{properties}. Their generalization
to another problem requires generalization of certain properties, which is
assumed to be made independently. It may also require changes in the path
of the proof, and even adding new lemmas. The mathematical concepts are
common in the field of partial differential equations: geometric domains,
variables defined on these domains, functions of several variables,
operators (e.g. derivatives, integrals, two-scale transform, etc..).
Finally, the proofs of Lemmas are designed to be realizable by rewriting.

\bigskip

Then, we presents a computational framework based on the theory of rewriting
to express the above method. Each \textit{property} is expressed as a
rewrite rule that can be conditional, so that it can be applied or not
according to a given logical formula. A step in a lemma proof is realized by
a strategy that expresses how the rule applies. The complete proof of a
lemma is then a sequence of such strategies. Ones we use have been developed
in previous work \cite{BGL-JSC10} that is implemented in Maple$^{\registered %
}$, here we provide its formalization. To allow the successful application
of rewriting strategies to an expression that contains associative and/or
commutative operations, such as $+,*,\cup,\cap$, etc, we use the concept of
rewriting modulo an equational theory \cite[\S 11]{BaaNip}. Without such
concept one needs to duplicate the rewriting rules.

In this work, rewriting operates on expressions whose level of abstraction
accurately reflects the mathematical framework. Concrete descriptions of
geometric domains, functions or operators are not provided. Their
description follows a grammar that has been defined in order that they carry
enough information allowing for the design of the rewriting rules and the
strategies. In some conditions of rewriting rules, the set of variables on
which an expression depends is required. This is for example the case for
the linearity property of the integral. Rather than introducing a typing
system, which would be cumbersome and restrictive, we introduced a specific
functionality in the form of a $\lambda $-term (i.e. a program). The
language of strategy allows this use. Put together all these concepts can
express a lemma proof as a strategy, i.e. a first order strategy, and
therefore provide a framework of symbolic computation. The concept of
generalization of a proof is introduced as second order rewrite strategies,
made with second order rewriting rules, operating on first order strategies.
They can transform first order rewrite rules and strategies and, where
appropriate, remove or add new ones. This framework has been implemented in
the software Maple$^{\registered }$. We present its application to the
complete proof of the reference problem and also to the generalizations of
the first lemma, by applying second order strategies, to multi-dimensional
geometrical domains, multi-dimensional thin domains and multi-domains.

The paper is organized as follows. Section \ref{Skeleton:Sec} is devoted to
all mathematical aspects. This includes all definitions and properties, the
lemmas and their proof. The principles of rewrite rules and strategies are
formulated in Section \ref{rewriting_stratgeies:sec}. Section \ref{A
framework for modelderivation} is devoted to the theoretical framework that
allows to derive a model and its generalizations. Implementation results are
described in section \ref{Implementation and Experiments}.


\section{Skeleton of two-scale modeling\label{Skeleton:Sec}}

We recall the framework of the two-scale convergence as presented in \cite%
{LenSmi07}, and the proof of the \textit{reference model} whose
implementation and extension under the form of algorithms of symbolic
computation are discussed in Section \ref{Implementation and Experiments}.
The presentation is divided into three subsections.\ The first one is
devoted to basic definitions and properties, stated as \textit{Propositions}%
. The latter are admitted without proof because they are assumed to be
prerequisites, or building blocks, in the proofs. They are used as
elementary steps in the two other sections detailing the proof of the
convergence of the two-scale transform of a derivative, and the homogenized
model derivation. The main statements of these two subsections are also
stated as \textit{Propositions} and their proofs are split into numbered
blocks called lemmas. Each lemma is decomposed into steps refering to the
definitions and propositions. All components of the \textit{reference model }%
derivation, namely the definitions, the propositions, the lemmas and the
proof steps are designed so that to be easily implemented and also to be
generalized for more complex models. We quote that a number of elementary
properties are used in the proof but are not explicitely stated nor cited.

\subsection{Notations, Definitions and Propositions\label{Notations,
Definitions and Propositions}}

\noindent Note that the functional framework used in this section is not as
precise as it should be for a usual mathematical work. The reason is that
the functional analysis is not covered by our symbolic computation.\ So,
precise mathematical statements and justifications are not in the focus of
this work.

\bigskip

\noindent In the sequel, $A\subset \mathbb{R}^{n}$ is a bounded open set$,$
with measure $|A|,$ having a "sufficiently"\textit{\ }regular boundary $%
\partial A$ and with unit outward normal denoted by $n_{\partial A}$. We
shall use the set $L^{1}(A)$ of integrable functions and the set $L^{p}(A)$,
for any $p>0$, of functions $f$ such that $f^{p}\in L^{1}(A),$ with norm $%
||v||_{L^{p}(A)}=(\int_{A}|v|^{p}$ $dx)^{1/p}.$ The Sobolev space $H^{1}(A)$
is the set of functions $f\in L^{2}(A)$ whose  gradient $\nabla f\in
L^{2}(A)^{n}.$ The set of $p$ times differentiable functions on $A$ is
denoted by $\mathcal{C}^{p}(A)$, where $p$ can be any integer or $\infty $.
Its subset $\mathcal{C}_{0}^{p}(A)$ is composed of functions whose  partial
derivatives are vanishing on the boundary $\partial A$ of $A$ until the
order $p$. For any integers $p$ and $q,$ $\mathcal{C}^{q}(A)\subset L^{p}(A)$%
. When $A=(0,a_{1})\times ...\times (0,a_{n})$ is a cuboid (or rectangular
parallelepiped) we say that a function $v$ defined in $\mathbb{R}^{n}$ is $A$%
-periodic if for any $\ell \in \mathbb{Z}^{n},$ $v(y+\sum_{i=1}^{n}\ell
_{i}a_{i}e_{i})=v(y)$ where $e_{i}$ is the $i^{th}$ vector of the canonical
basis of $\mathbb{R}^{n}$. The set of $A$-periodic functions which are $%
\mathcal{C}^{\infty }$ is denoted by $\mathcal{C}_{\sharp }^{\infty }(A)$
and those which are in $H^{1}(A)$ is denoted by $H_{\sharp }^{1}(A)$. The
operator $tr$ (we say \textit{trace}) can be defined as the restriction
operator from functions defined on the closure of $A$ to functions defined
on its boundary $\partial A$. Finally, we say that a sequence $%
(u^{\varepsilon })_{\varepsilon >0}\in L^{2}(A)$ converges strongly in $%
L^{2}(A)$ towards $u^{0}\in L^{2}(A)$ when $\varepsilon $ tends to zero if $%
\lim_{\varepsilon \rightarrow 0}||u^{\varepsilon }-u^{0}||_{L^{2}(A)}=0$.
The convergence is said to be weak if $\lim_{\varepsilon \rightarrow
0}\int_{A}(u^{\varepsilon }-u^{0})v$ $dx=0$ for all $v\in L^{2}(A)$. We
write $u^{\varepsilon }=u^{0}+O_{s}(\varepsilon )$ (respectively $%
O_{w}(\varepsilon )$)$,$ where $O_{s}(\varepsilon )$ (respectively $%
O_{w}(\varepsilon )$) represents a sequence tending to zero strongly
(respectively weakly) in $L^{2}(A)$. Moreover, the simple notation $%
O(\varepsilon )$ refers to a sequence of numbers which simply tends to zero.
We do not detail the related usual computation rules.

\begin{proposition}
\label{Interpretation of a weak equality}\textbf{[Interpretation of a weak
equality] }For $u\in L^{2}(A)$ and for any $v\in \mathcal{C}_{0}^{\infty }(A)$,%
\begin{equation*}
\text{if }\int_{A}u(x)\text{ }v(x)\text{ }dx=0\text{ then }u=0
\end{equation*}%
in the sense of $L^{2}(A)$ functions.
\end{proposition}

\begin{proposition}
\label{Interpretation of a periodic boundary condition}\textbf{%
[Interpretation of a periodic boundary condition] }For $u\in H^{1}(A)$ and
for any $v\in \mathcal{C}_{\#}^{\infty }\left( A\right) $,%
\begin{equation*}
\text{if }\int_{\partial A}u(x)\ v(x)\ n_{\partial A}(x)\ dx=0\text{ then }%
u\in H_{\sharp }^{1}\left( A\right) .
\end{equation*}
\end{proposition}

\noindent In the remainder of this section, only the dimension $n=1$ is
considered, the general definitions being used for the generalizations
discussed in Section \ref{Implementation and Experiments}.

\begin{notation}
\label{Physical and microscopic Domains}\textbf{[Physical and microscopic
Domains] }We consider an interval $\Omega
=\bigcup\limits_{c=1}^{N(\varepsilon )}\Omega _{c}^{1,\varepsilon }\subset
\mathbb{R}$ divided into $N(\varepsilon )$ periodic cells (or intervals) $%
\Omega _{c}^{1,\varepsilon }$, of size $\varepsilon >0$, indexed by $c$, and
with center $x_{c}.$ The translation and magnification $(\Omega
_{c}^{1,\varepsilon }-x_{c})/\varepsilon $ is called the \textit{unit cell }%
and is denoted by $\Omega ^{1}$. The variables in $\Omega $ and in $\Omega
^{1}$ are denoted by $x^{\varepsilon }$ and $x^{1}.$
\end{notation}

\noindent The two-scale transform $T$ is an operator mapping functions
defined in the physical domain $\Omega $ to functions defined in the
two-scale domain $\Omega ^{\sharp }\times \Omega ^{1}$ where for the \textit{%
reference model} $\Omega ^{\sharp }=\Omega $. In the following, we shall
denote by $\Gamma ,$ $\Gamma ^{\sharp }$ and $\Gamma ^{1}$ the boundaries of
$\Omega ,$ $\Omega ^{\sharp }$ and $\Omega ^{1}$.

\begin{definition}
\label{Two-Scale Transform}\textbf{[Two-Scale Transform]} The two-scale
transform $T$ is the linear operator defined by
\begin{equation}
(Tu)(x_{c},x^{1})=u(x_{c}+\varepsilon x^{1})  \label{Def_TS_1D:nosum}
\end{equation}%
and then by extension $T(u)(x^{\sharp },x^{1})=u(x_{c}+\varepsilon x^{1})$
for all $x^{\sharp }\in \Omega _{c}^{1,\varepsilon }$ and each $c$ in $%
1,..,N(\varepsilon ).$
\end{definition}

\begin{notation}
\label{Measure of Domains}\textbf{[Measure of Domains] }$\kappa ^{0}=\frac{1%
}{|\Omega |}$ and $\kappa ^{1}=\frac{1}{|\Omega ^{\sharp }\times \Omega ^{1}|%
}.$
\end{notation}

\noindent The operator $T$ enjoys the following properties.

\begin{proposition}
\label{Product Rule}\textbf{[Product Rule] }For two functions $u,$ $v$
defined in $\Omega ,$%
\begin{equation}
T(uv)=(Tu)(Tv).  \label{MultiRuleOfT}
\end{equation}
\end{proposition}

\begin{proposition}
\label{Derivative Rule}\textbf{[Derivative Rule] }If $u$ and its derivative
are defined in $\Omega $ then%
\begin{equation}
T\left( \frac{du}{dx}\right) =\frac{1}{\varepsilon }\frac{\partial (Tu)}{%
\partial x^{1}}.  \label{PartialRuleOfT}
\end{equation}
\end{proposition}

\begin{proposition}
\label{Integral Rule}\textbf{[Integral Rule] }If a function $u\in
L^{1}(\Omega )$ then $Tu\in L^{1}(\Omega ^{\sharp }\times \Omega ^{1})$ and%
\begin{equation}
\kappa ^{0}\int_{\Omega }u\text{ }dx=\kappa ^{1}\int_{\Omega ^{\sharp
}\times \Omega ^{1}}(Tu)\text{ }dx^{\sharp }dx^{1}.
\label{IntegralPropertyOfT}
\end{equation}
\end{proposition}

\noindent The next two properties are corollaries of the previous ones.

\begin{proposition}
\label{Inner Product Rule}\textbf{[Inner Product Rule] }For two functions $%
u, $ $v\in L^{2}(\Omega ),$%
\begin{equation}
\kappa ^{0}\int_{\Omega }u\text{ }v\text{ }dx=\kappa ^{1}\int_{\Omega
^{\sharp }\times \Omega ^{1}}(Tu)\text{ }(Tv)\text{ }dx^{\sharp }dx^{1}.
\end{equation}
\end{proposition}

\begin{proposition}
\label{Norm Rule}\textbf{[Norm Rule] }For a function $u\in L^{2}(\Omega ),$%
\begin{equation}
\kappa ^{0}\left\Vert u\right\Vert _{L^{2}(\Omega )}^{2}=\kappa
^{1}\left\Vert Tu\right\Vert _{L^{2}(\Omega ^{\sharp }\times \Omega
^{1})}^{2}.  \label{NormRuleOfT}
\end{equation}
\end{proposition}

\begin{definition}
\label{Two-Scale Convergence}\textbf{[Two-Scale Convergence] }A sequence $%
u^{\varepsilon }\in L^{2}(\Omega )$ is said to be two-scale strongly
(respect. weakly) convergent in $L^{2}(\Omega ^{\sharp }\times \Omega ^{1})$
to a limit $u^{0}(x^{\sharp },x^{1})$ if $Tu^{\varepsilon }$ is strongly
(respect. weakly) convergent towards $u^{0}$ in $L^{2}(\Omega ^{\sharp
}\times \Omega ^{1}).$
\end{definition}

\begin{definition}
\label{Adjoint or Dual of T}\textbf{[Adjoint or Dual of T] }As $T$ is a
linear operator from $L^{2}(\Omega )$ to $L^{2}(\Omega ^{\sharp }\times
\Omega ^{1}),$ its adjoint $T^{\ast }$ is a linear operator from $%
L^{2}(\Omega ^{\sharp }\times \Omega ^{1})$ to $L^{2}(\Omega )$ defined by
\begin{equation}
\kappa ^{0}\int_{\Omega }T^{\ast }v\text{ }u\text{ }dx=\kappa
^{1}\int_{\Omega ^{\sharp }\times \Omega ^{1}}v\text{ }Tu\text{ }dx^{\sharp
}dx^{1}.  \label{Def_TSAdj_1D}
\end{equation}
\end{definition}

\noindent The expression of $T^{\ast }$ can be explicited, it maps regular
functions in $\Omega ^{\sharp }\times \Omega ^{1}$ to piecewise-constant
functions in $\Omega $. The next definition introduce an operator used as a
smooth approximation of $T^{\ast }$.

\begin{definition}
\label{Regularization of T*}\textbf{[Regularization of T}$^{\ast }$\textbf{]
}The operator $B$ is the linear continuous operator defined from $%
L^{2}(\Omega ^{\sharp }\times \Omega ^{1})$ to $L^{2}(\Omega )$ by%
\begin{equation}
Bv=v(x,\frac{x}{\varepsilon }).  \label{Def_Bar_1D}
\end{equation}
\end{definition}

\noindent The nullity condition of a function $v(x^{\sharp },x^{1})$ on the
boundary $\partial \Omega ^{\sharp }\times \Omega ^{1}$ is transferred to
the range $Bv$ as follows.

\begin{proposition}
\label{Boundary Conditions of Bv}\textbf{[Boundary Conditions of Bv] }If $%
v\in \mathcal{C}_{0}^{\infty }(\Omega ^{\sharp };\mathcal{C}^{\infty
}(\Omega ^{1}))$ then $Bv\in \mathcal{C}_{0}^{\infty }(\Omega )$.
\end{proposition}

\begin{proposition}
\label{Derivation Rule for B}\textbf{[Derivation Rule for B] }If $v$ and its
partial derivatives are defined on $\Omega ^{\sharp }\times \Omega ^{1}$ then%
\begin{equation}
\frac{d(Bv)}{dx}=B(\frac{\partial v}{\partial x^{\sharp }})+\varepsilon
^{-1}B(\frac{\partial v}{\partial x^{1}}).  \label{partial:Bv:skeleton}
\end{equation}
\end{proposition}

\noindent The next proposition states that the operator $B$ is actually an
approximation of the operator $T^{\ast }$ for $\Omega ^{1}$-periodic
functions.

\begin{proposition}
\label{Approximation between T* by B}\textbf{[Approximation between T}$%
^{\ast }$ \textbf{and B] }If $v(x^{\sharp },x^{1})$ is continuous,
continuously differentiable in $x^{\sharp }$ and $\Omega ^{1}$-periodic in $%
x^{1}$ then%
\begin{equation}
T^{\ast }v=Bv-\varepsilon B(x^{1}\frac{\partial v}{\partial x^{\sharp }}%
)+\varepsilon O_{s}(\varepsilon ).  \label{Inversion_Formula_1D}
\end{equation}%
Conversely,%
\begin{equation}
Bv=T^{\ast }(v)+\varepsilon T^{\ast }(x^{1}\frac{\partial v}{\partial
x^{\sharp }})+\varepsilon O_{s}(\varepsilon ).
\label{First_Order_Approx_of_Bar}
\end{equation}
\end{proposition}

\noindent Next, the formula of integration by parts is stated in a form
compatible with the Green formula used in some extensions. The boundary $%
\Gamma $ is composed of the two end points of the interval $\Omega $, and
the unit outward normal $n_{\Gamma }$ defined on $\Gamma $ is equal to $-1$
and $+1$ at the left- and right-endpoints respectively.

\begin{proposition}
\label{Green Rule}\textbf{[Green Rule] }If $u$, $v\in H^{1}(\Omega )$ then
the traces of $u$ and $v$ on $\Gamma $ are well defined and%
\begin{equation}
\int_{\Omega }u\frac{dv}{dx}\text{ }dx=\int_{\Gamma }tr(u)\text{ }tr(v)\text{
}n_{\Gamma }\text{ }ds(x)-\int_{\Omega }v\frac{du}{dx}\text{ }dx.
\label{GreenFormulas}
\end{equation}
\end{proposition}

\noindent The last proposition is stated as a building block of the
homogenized model derivation.

\begin{proposition}
\label{The linear operator associated to the Microscopic problem}\textbf{%
[The linear operator associated to the Microscopic problem] }For $\mu \in
\mathbb{R}$, there exist $\theta ^{\mu }\in H_{\sharp }^{1}(\Omega ^{1})$
solutions to the linear weak formulation%
\begin{equation}
\int_{\Omega ^{1}}a^{0}\frac{\partial \theta ^{\mu }}{\partial x^{1}}\frac{%
\partial w}{\partial x^{1}}\text{ }dx^{1}=-\mu \int_{\Omega ^{1}}a^{0}\frac{%
\partial w}{\partial x^{1}}\text{ }dx^{1}\text{ for all }w\in \mathcal{C}%
_{\sharp }^{\infty }(\Omega ^{1}),  \label{Microscopic problem}
\end{equation}%
and $\frac{\partial \theta ^{\mu }}{\partial x^{1}}$ is unique. Since the
mapping $\mu \mapsto \dfrac{\partial \theta ^{\mu }}{\partial x^{1}}$ from $%
\mathbb{R}$ to $L^{2}(\Omega ^{1})$ is linear then
\begin{equation}
\dfrac{\partial \theta ^{\mu }}{\partial x^{1}}=\mu \dfrac{\partial \theta
^{1}}{\partial x^{1}}.  \label{microscopic linear operator}
\end{equation}%
Moreover, this relation can be extended to any $\mu \in L^{2}(\Omega
^{\sharp })$.
\end{proposition}

\subsection{Two-Scale Approximation of a Derivative \label{Two-Scale
Approximation of a Derivative}}

\noindent Here we detail the \textit{reference computation} of the weak
two-scale limit $\eta =\lim_{\varepsilon \rightarrow 0}T(\frac{%
du^{\varepsilon }}{dx})$ in $L^{2}(\Omega ^{\sharp }\times \Omega ^{1})$ when%
\begin{equation}
\left\Vert u^{\varepsilon }\right\Vert _{L^{2}(\Omega )}\text{ and }%
\left\Vert \frac{du^{\varepsilon }}{dx}\right\Vert _{L^{2}(\Omega )}\leq C,
\label{Boundness of u}
\end{equation}%
$C$ being a constant independent of $\varepsilon $. To simplify the proof,
we further assume that there exist $u^{0}$, $u^{1}\in L^{2}(\Omega ^{\sharp
}\times \Omega ^{1})$ such that%
\begin{equation*}
T(u^{\varepsilon })=u^{0}+\varepsilon u^{1}+\varepsilon O_{w}(\varepsilon )%
\text{,}
\end{equation*}%
i.e.%
\begin{equation}
\int_{\Omega ^{\sharp }\times \Omega ^{1}}(T(u^{\varepsilon
})-u^{0}-\varepsilon u^{1})v\text{ }dx^{\sharp }dx^{1}=\varepsilon
O(\varepsilon )\text{ for all }v\in L^{2}(\Omega ^{\sharp }\times \Omega
^{1}).  \label{Tu:expression:Skel}
\end{equation}%
We quote that Assumption (\ref{Tu:expression:Skel}) is not necessary, it is
introduced to simplify the proof since it avoids some non-equational steps.
The statement proved in the remaining of the subsection is the following.

\begin{proposition}
\label{Two-scale Limit of a Derivative}\textbf{[Two-scale Limit of a
Derivative] }If $u^{\varepsilon }$ is a sequence bounded as in (\ref%
{Boundness of u}) and satisfying (\ref{Tu:expression:Skel}), then $u^{0}$ is
independent of $x^{1},$
\begin{equation}
\tilde{u}^{1}=u^{1}-x^{1}\partial _{x^{\sharp }}u^{0}
\label{Def:Tilde:u:skeleton}
\end{equation}%
defined in $\Omega ^{\sharp }\times \Omega ^{1}$ is $\Omega ^{1}$-periodic
and%
\begin{equation}
\eta =\frac{\partial u^{0}}{\partial x^{\sharp }}+\frac{\partial \tilde{u}%
^{1}}{\partial x^{1}}.  \label{WeakLimitation}
\end{equation}%
Moreover, if $u^{\varepsilon }=0$ on $\Gamma $ then $u^{0}=0$ on $\Gamma
^{\sharp }.$
\end{proposition}

\noindent The proof is split into four Lemmas corresponding to the first
four blocks discussed in Section \ref{Implementation and Experiments}, the
other three being detailed in subsection \ref{Homogenized Model Derivation}.

\begin{lemma}
\label{First Block}\textbf{[First Block: Constraint on }$u^{0}$\textbf{] }$%
u^{0}$ is independent of $x^{1}$.
\end{lemma}

\begin{proof}
We introduce%
\begin{equation*}
\Psi =\varepsilon \kappa ^{0}{\int_{\Omega }\frac{du^{\varepsilon }}{dx}Bv}%
\text{ }{dx}
\end{equation*}
with $v\in \mathcal{C}_{0}^{\infty }(\Omega ^{\sharp };\mathcal{C}%
_{0}^{\infty }(\Omega ^{1}))$. From the Cauchy-Schwartz inequality and (\ref%
{Boundness of u})$,$ $\lim_{\varepsilon \rightarrow 0}\Psi =0$.

\begin{itemize}
\item \textbf{Step 1. }The Green formula (\ref{GreenFormulas}) and
Proposition \ref{Boundary Conditions of Bv} $\Longrightarrow $%
\begin{equation*}
\Psi ={-\varepsilon \kappa ^{0}\int_{\Omega }u^{\varepsilon }\frac{d(Bv)}{dx}%
\text{ }dx.}
\end{equation*}

\item \textbf{Step 2. }Proposition \ref{Derivation Rule for B} $%
\Longrightarrow $%
\begin{equation*}
\Psi =\kappa ^{0}{\int_{\Omega }u^{\varepsilon }B(\frac{\partial v}{\partial
x^{1}})\text{ }dx+O(\varepsilon ).}
\end{equation*}

\item \textbf{Step 3. }Proposition \ref{Approximation between T* by B} $%
\Longrightarrow $%
\begin{equation*}
\Psi =\kappa ^{0}{\int_{\Omega }u^{\varepsilon }T^{\ast }(\frac{\partial v}{%
\partial x^{1}})}\text{ }{\text{ }dx+O(\varepsilon ).}
\end{equation*}

\item \textbf{Step 4. }Definition \ref{Adjoint or Dual of T} $%
\Longrightarrow $%
\begin{equation*}
\Psi =\kappa ^{1}{\int_{\Omega ^{\sharp }\times \Omega ^{1}}T(u^{\varepsilon
})\frac{\partial v}{\partial x^{1}}\text{ }dx+O(\varepsilon ).}
\end{equation*}

\item \textbf{Step 5. }Assumption (\ref{Tu:expression:Skel}) and passing to
the limit when $\varepsilon \rightarrow 0$ $\Longrightarrow $%
\begin{equation*}
\kappa ^{1}{\int_{\Omega ^{\sharp }\times \Omega ^{1}}u^{0}\frac{\partial v}{%
\partial x^{1}}\text{ }dx=0}.
\end{equation*}

\item \textbf{Step 6. }The Green formula (\ref{GreenFormulas}) and $v=0$ on $%
\Omega ^{\sharp }\times \Gamma ^{1}$ $\Longrightarrow $%
\begin{equation*}
\kappa ^{1}{\int_{\Omega ^{\sharp }\times \Omega ^{1}}\frac{\partial u^{0}}{%
\partial x^{1}}}\text{ }{v\text{ }dx=0}.
\end{equation*}

\item \textbf{Step 7. }Proposition \ref{Interpretation of a weak equality} $%
\Longrightarrow $%
\begin{equation*}
\dfrac{\partial u^{0}}{\partial x^{1}}=0.
\end{equation*}
\end{itemize}
\end{proof}

\begin{lemma}
\label{Second Block}\textbf{[Second Block: Two-Scale Limit of the
Derivative] }$\eta =\frac{\partial u^{1}}{\partial x^{1}}.$
\end{lemma}

\begin{proof}
We choose $v\in \mathcal{C}_{0}^{\infty }(\Omega ^{\sharp };\mathcal{C}%
_{0}^{\infty }(\Omega ^{1}))$ in%
\begin{equation}
\Psi =\kappa ^{1}\int_{\Omega ^{\sharp }\times \Omega ^{1}}T(\frac{%
du^{\varepsilon }}{dx})v\text{ }dx^{\sharp }dx^{1}.
\label{Initial term of the second block}
\end{equation}

\begin{itemize}
\item \textbf{Step 1.} Definition \ref{Adjoint or Dual of T} $%
\Longrightarrow $%
\begin{equation*}
\Psi =\kappa ^{0}\int_{\Omega }\frac{du^{\varepsilon }}{dx}T^{\ast }v\text{ }%
dx.
\end{equation*}

\item \textbf{Step 2.} Proposition {\ref{Approximation between T* by B}} (to
approximate $T^{\ast }$ by $B$), the Green formula ({\ref{GreenFormulas}}),
the linearity of integrals, and again Proposition {\ref{Approximation
between T* by B}} (to approximate $B$ by $T^{\ast }$) $\Longrightarrow $%
\begin{equation*}
\Psi =-\kappa ^{0}\int_{\Omega }u^{\varepsilon }T^{\ast }(\frac{\partial v}{%
\partial x^{\sharp }})\text{ }dx-\frac{\kappa ^{0}}{\varepsilon }%
\int_{\Omega }u^{\varepsilon }T^{\ast }(\frac{\partial v}{\partial x^{1}})%
\text{ }dx-\kappa ^{0}\int_{\Omega }u^{\varepsilon }T^{\ast }(\frac{\partial
^{2}v}{\partial x^{1}\partial x^{\sharp }}x^{1})\text{ }dx+O(\varepsilon ).
\end{equation*}

\item \textbf{Step 3.} Definition \ref{Adjoint or Dual of T} $%
\Longrightarrow $%
\begin{eqnarray*}
\Psi &=&-\kappa ^{1}\int_{\Omega ^{\sharp }\times \Omega
^{1}}T(u^{\varepsilon })\frac{\partial v}{\partial x^{\sharp }}\text{ }%
dx^{\sharp }dx^{1}-\frac{\kappa ^{1}}{\varepsilon }\int_{\Omega ^{\sharp
}\times \Omega ^{1}}T(u^{\varepsilon })\frac{\partial v}{\partial x^{1}}%
\text{ }dx^{\sharp }dx^{1} \\
&&-\kappa ^{1}\int_{\Omega ^{\sharp }\times \Omega ^{1}}T(u^{\varepsilon
})x^{1}\frac{\partial ^{2}v}{\partial x^{1}\partial x^{\sharp }}\text{ }%
dx^{\sharp }dx^{1}+O(\varepsilon ).
\end{eqnarray*}

\item \textbf{Step 4. }Assumption (\ref{Tu:expression:Skel}) $%
\Longrightarrow $%
\begin{eqnarray*}
\Psi &=&-\kappa ^{1}\int_{\Omega ^{\sharp }\times \Omega ^{1}}u^{0}\frac{%
\partial v}{\partial x^{\sharp }}\text{ }dx^{\sharp }dx^{1}-\frac{\kappa ^{1}%
}{\varepsilon }\int_{\Omega ^{\sharp }\times \Omega ^{1}}u^{0}\frac{\partial
v}{\partial x^{1}}\text{ }dx^{\sharp }dx^{1}-\kappa ^{1}\int_{\Omega
^{\sharp }\times \Omega ^{1}}u^{1}\frac{\partial v}{\partial x^{1}}\text{ }%
dx^{\sharp }dx^{1} \\
&&-\kappa ^{1}\int_{\Omega ^{\sharp }\times \Omega ^{1}}u^{0}\frac{\partial
^{2}v}{\partial x^{1}\partial x^{\sharp }}x^{1}+O(\varepsilon ).
\end{eqnarray*}

\item \textbf{Step 5.} The Green formula (\ref{GreenFormulas}), Lemma \ref%
{First Block}, and passing to the limit when $\varepsilon \rightarrow 0$ $%
\Longrightarrow $%
\begin{equation*}
\kappa ^{1}\int_{\Omega ^{\sharp }\times \Omega ^{1}}\eta \text{ }v\text{ }%
dx^{\sharp }dx^{1}=\kappa ^{1}\int_{\Omega ^{\sharp }\times \Omega ^{1}}%
\frac{\partial u^{1}}{\partial x^{1}}v\text{ }dx^{\sharp }dx^{1}.
\end{equation*}

\item \textbf{Step 6.} Proposition \ref{Interpretation of a weak equality} $%
\Longrightarrow $%
\begin{equation*}
\eta =\frac{\partial u^{1}}{\partial x^{1}}.
\end{equation*}
\end{itemize}
\end{proof}

\begin{lemma}
\label{Third Block}\textbf{[Third Block: Microscopic Boundary Condition] }$%
\tilde{u}^{1}$ is $\Omega ^{1}$-periodic.
\end{lemma}

\begin{proof}
In (\ref{Initial term of the second block}), we choose $v\in \mathcal{C}%
_{0}^{\infty }(\Omega ^{\sharp };\mathcal{C}_{\sharp }^{\infty }(\Omega
^{1}))$.

\begin{itemize}
\item \textbf{Step 1.} The steps 1-5 of the second block $\Longrightarrow $%
\begin{equation*}
\kappa ^{1}\int_{\Omega ^{\sharp }\times \Omega ^{1}}\eta v\text{ }%
dx^{\sharp }dx^{1}-\kappa ^{1}\int_{\Omega ^{\sharp }\times \Gamma
^{1}}(u^{1}-x^{1}\frac{\partial u^{0}}{\partial x^{\sharp }})v\text{ }%
n_{\Gamma ^{1}}\text{ }dx^{\sharp }dx^{1}-\kappa ^{1}\int_{\Omega ^{\sharp
}\times \Omega ^{1}}\frac{\partial u^{1}}{\partial x^{1}}v\text{ }dx^{\sharp
}dx^{1}=0.
\end{equation*}

\item \textbf{Step 2.} Lemma \ref{Second Block} $\Longrightarrow $%
\begin{equation}
\int_{\Omega ^{\sharp }\times \Gamma ^{1}}(u^{1}-x^{1}\frac{\partial u^{0}}{%
\partial x^{\sharp }})v\text{ }n_{\Gamma ^{1}}\text{ }dx^{\sharp
}ds(x^{1})=0.  \label{PeriodicPre}
\end{equation}

\item \textbf{Step 3.} Definition (\ref{Def:Tilde:u:skeleton}) of $\tilde{u}%
^{1}$ and Proposition \ref{Interpretation of a periodic boundary condition} $%
\Longrightarrow $
\begin{equation}
\tilde{u}^{1}\text{ is }\Omega ^{1}\text{-periodic.}  \label{Conclusion2}
\end{equation}
\end{itemize}
\end{proof}

\begin{lemma}
\label{Fourth Block}\textbf{[Fourth Block: Macroscopic Boundary Condition] }$%
u^{0}$ vanishes on $\Gamma ^{\sharp }$.
\end{lemma}

\begin{proof}
We choose $v\in \mathcal{C}_{0}^{\infty }(\Omega ^{\sharp })$,

\begin{itemize}
\item \textbf{Step 1.} The steps 1-5 of the second block and $u^{\varepsilon
}=0$ on $\Gamma $ $\Longrightarrow $%
\begin{equation*}
\int_{\Gamma ^{\sharp }\times \Omega ^{1}}u^{0}v\text{ }n_{\Gamma ^{\sharp }}%
\text{ }ds(x^{\sharp })dx^{1}=0.
\end{equation*}

\item \textbf{Step 2.} Proposition \ref{Interpretation of a weak equality} $%
\Longrightarrow $
\begin{equation*}
u^{0}=0\text{ on }\Gamma ^{\sharp }.
\end{equation*}
\end{itemize}
\end{proof}

\subsection{Homogenized Model Derivation \label{Homogenized Model Derivation}%
}

\noindent Here we provide the \textit{reference proof }of the homogenized
model derivation. It uses Proposition \ref{Two-scale Limit of a Derivative}
as an intermediary result. Let $u^{\varepsilon }$, the solution of a linear
boundary value problem posed in $\Omega ,$%
\begin{equation}
\left\{
\begin{array}{l}
-\dfrac{d}{dx}(a^{\varepsilon }(x)\dfrac{du^{\varepsilon }(x)}{dx})=f\text{
in }\Omega \\
u^{\varepsilon }=0\text{ on }\Gamma ,%
\end{array}%
\right.  \label{Original Model}
\end{equation}%
where the right-hand side $f\in L^{2}(\Omega ),$ the coefficient $%
a^{\varepsilon }\in \mathcal{C}^{\infty }(\Omega )$ is $\varepsilon \Omega
^{1}$-periodic, and there exist two positive constants $\alpha $ and $\beta $
independent $\varepsilon $ such that%
\begin{equation}
0<\alpha \leq a^{\varepsilon }(x)\leq \beta .
\label{Estimate of Coefficient}
\end{equation}%
The weak formulation is obtained by multiplication of the differential
equation by a test function $v\in \mathcal{C}_{0}^{\infty }(\Omega )$ and
application of the Green formula,%
\begin{equation}
\kappa ^{0}\int_{\Omega }a^{\varepsilon }(x)\frac{du^{\varepsilon }}{dx}%
\frac{dv}{dx}\text{ }dx=\kappa ^{0}\int_{\Omega }f(x)v(x)\text{ }dx.
\label{Model1}
\end{equation}%
It is known that its unique solution $u^{\varepsilon }$ is bounded as in (%
\ref{Boundness of u}). Moreover, we assume that for some functions $%
a^{0}(x^{1})$ and $f^{0}(x^{\sharp }),$
\begin{equation}
T(a^{\varepsilon })=a^{0}\text{ and }T(f)=f^{0}(x^{\sharp
})+O_{w}(\varepsilon ).  \label{T(a) and T(f)}
\end{equation}%
\noindent The next proposition states the homogenized model and is the main
result of the \textit{reference proof}. For $\theta ^{1}$ a solution to the
microscopic problem (\ref{Microscopic problem}) with $\mu =1,$ the
homogenized coefficient and right-hand side are defined by%
\begin{equation}
a^{H}=\int_{\Omega ^{1}}a^{0}\left( 1+\frac{\partial \theta ^{1}}{\partial
x^{1}}\right) ^{2}\text{ }dx^{1}\text{ and }f^{H}=\int_{\Omega ^{1}}f^{0}%
\text{ }dx^{1}.  \label{coeff homogeneises}
\end{equation}

\begin{proposition}
\label{Homogenized Model}\textbf{[Homogenized Model] }The limit $u^{0}$ is
solution to the weak formulation
\begin{equation}
\int_{\Omega ^{\sharp }}a^{H}\frac{du^{0}}{dx^{\sharp }}\frac{dv^{0}}{%
dx^{\sharp }}\text{ }dx^{\sharp }=\int_{\Omega ^{\sharp }}f^{H}v^{0}\text{ }%
dx^{\sharp }  \label{Homogenized model weak form}
\end{equation}%
for all $v^{0}\in \mathcal{C}_{0}^{\infty }(\Omega ^{\sharp }).$
\end{proposition}

\noindent The proof is split into three lemmas.

\begin{lemma}
\label{Fifth Block}\textbf{[Fifth Block: Two-Scale Model] }The couple $%
(u^{0},\widetilde{u}^{1})$ is solution to the two-scale weak formulation%
\begin{equation}
\int_{\Omega ^{\sharp }\times \Omega ^{1}}a^{0}\left( \frac{\partial u^{0}}{%
\partial x^{\sharp }}+\frac{\partial \widetilde{u}^{1}}{\partial x^{1}}%
\right) \left( \frac{\partial v^{0}}{\partial x^{\sharp }}+\frac{\partial
v^{1}}{\partial x^{1}}\right) \text{ }dx^{\sharp }dx^{1}=\int_{\Omega
^{\sharp }\times \Omega ^{1}}f^{0}v^{0}\text{ }dx^{\sharp }dx^{1}
\label{TwoScaleModel}
\end{equation}%
for any $v^{0}\in \mathcal{C}_{0}^{\infty }(\Omega ^{\sharp })$ and $%
v^{1}\in \mathcal{C}_{0}^{\infty }(\Omega ^{\sharp },C_{\sharp }^{\infty
}(\Omega ^{1})).$
\end{lemma}

\begin{proof}
\noindent We choose the test functions $v^{0}\in \mathcal{C}_{0}^{\infty
}(\Omega ^{\sharp })$, $v^{1}\in \mathcal{C}_{0}^{\infty }(\Omega ^{\sharp
},C_{\sharp }^{\infty }(\Omega ^{1}))$.

\begin{itemize}
\item \textbf{Step 1} Posing $v=B(v^{0}+\varepsilon v^{1})$ in (\ref{Model1}%
) and Proposition \ref{Boundary Conditions of Bv} $\Longrightarrow $\textbf{%
\ }%
\begin{equation*}
Bv\in \mathcal{C}_{0}^{\infty }(\Omega )\text{ and }\kappa ^{0}\int_{\Omega
}a^{\varepsilon }\frac{du^{\varepsilon }}{dx}\frac{dB(v^{0}+\varepsilon
v^{1})}{dx}\text{ }dx=\kappa ^{0}\int_{\Omega }f\text{ }B(v^{0}+\varepsilon
v^{1})\text{ }dx.
\end{equation*}

\item \textbf{Step 2} Propositions \ref{Derivation Rule for B} and \ref%
{Approximation between T* by B} $\Longrightarrow $%
\begin{equation*}
\kappa ^{0}\int_{\Omega }a^{\varepsilon }\frac{du^{\varepsilon }}{dx}T^{\ast
}\left( \frac{\partial v^{0}}{\partial x^{\sharp }}+\frac{\partial v^{1}}{%
\partial x^{1}}\right) dx=\kappa ^{0}\int_{\Omega }f\text{ }T^{\ast
}(v^{0})dx+O(\varepsilon ).
\end{equation*}

\item \textbf{Step 3} Definition \ref{Adjoint or Dual of T} and Proposition %
\ref{Product Rule} $\Longrightarrow $%
\begin{equation}
\kappa ^{1}\int_{\Omega ^{\sharp }\times \Omega ^{1}}T(a^{\varepsilon })T(%
\frac{du^{\varepsilon }}{dx})\left( \frac{\partial v^{0}}{\partial x^{\sharp
}}+\frac{\partial v^{1}}{\partial x^{1}}\right) \text{ }dx^{\sharp
}dx^{1}=\kappa ^{1}\int_{\Omega ^{\sharp }\times \Omega ^{1}}T(f)\text{ }%
v^{0}\text{ }dx^{\sharp }dx^{1}+O(\varepsilon ).  \label{Model2}
\end{equation}

\item \textbf{Step 4} Definitions (\ref{T(a) and T(f)}), Lemma \ref%
{Two-scale Limit of a Derivative}, and passing to the limit when $%
\varepsilon \rightarrow 0$ $\Longrightarrow $%
\begin{equation*}
\int_{\Omega ^{\sharp }\times \Omega ^{1}}a^{0}\left( \frac{\partial u^{0}}{%
\partial x^{\sharp }}+\frac{\partial \widetilde{u}^{1}}{\partial x^{1}}%
\right) \left( \frac{\partial v^{0}}{\partial x^{\sharp }}+\frac{\partial
v^{1}}{\partial x^{1}}\right) \text{ }dx^{\sharp }dx^{1}=\int_{\Omega
^{\sharp }\times \Omega ^{1}}f^{0}v^{0}\text{ }dx^{\sharp }dx^{1}
\end{equation*}%
which is the expected result.
\end{itemize}
\end{proof}

\begin{lemma}
\label{Sixth Block}\textbf{[Sixth Block: Microscopic Problem] }$\widetilde{u}%
^{1}$ is solution to (\ref{Microscopic problem}) with $\mu =\dfrac{\partial
u^{0}}{\partial x^{\sharp }}$ and
\begin{equation*}
\frac{\partial \widetilde{u}^{1}}{\partial x^{1}}=\dfrac{\partial u^{0}}{%
\partial x^{\sharp }}\frac{\partial \theta ^{1}}{\partial x^{1}}.
\end{equation*}
\end{lemma}

\begin{proof}
We choose $v^{0}=0$ and $v^{1}(x^{\sharp },x^{1})=w(x^{1})\varphi (x^{\sharp
})$ in (\ref{TwoScaleModel}) with $\varphi \in \mathcal{C}^{\infty }(\Omega
^{\sharp })$ and $w^{1}\in \mathcal{C}_{\sharp }^{\infty }(\Omega ^{1})$.

\begin{itemize}
\item \textbf{Step 1 }Proposition \ref{Interpretation of a weak equality},
Lemma \ref{First Block}, and the linearity of the integral $\Longrightarrow $%
\begin{equation}
\int_{\Omega ^{1}}a^{0}\frac{\partial \widetilde{u}^{1}}{\partial x^{1}}%
\frac{\partial w^{1}}{\partial x^{1}}\text{ }dx^{1}=-\frac{\partial u^{0}}{%
\partial x^{\sharp }}\int_{\Omega ^{1}}a^{0}\frac{\partial w^{1}}{\partial
x^{1}}\text{ }dx^{1}.  \label{MicroScopicProblem}
\end{equation}

\item \textbf{Step 2 }Proposition \ref{The linear operator associated to the
Microscopic problem} with $\mu =\dfrac{\partial u^{0}}{\partial x^{\sharp }}$
$\Longrightarrow $%
\begin{equation*}
\frac{\partial \widetilde{u}^{1}}{\partial x^{1}}=\dfrac{\partial u^{0}}{%
\partial x^{\sharp }}\frac{\partial \theta ^{1}}{\partial x^{1}}
\end{equation*}%
as announced.
\end{itemize}
\end{proof}

\begin{lemma}
\label{Seventh Block}\textbf{[Seventh Block: Macroscopic Problem] }$u^{0}$
is solution to (\ref{Homogenized model weak form}).
\end{lemma}

\begin{proof}
We choose $v^{0}\in \mathcal{C}_{0}^{\infty }(\Omega ^{\sharp })$ and $v^{1}=%
\dfrac{\partial v^{0}}{\partial x^{\sharp }}\dfrac{\partial \theta ^{1}}{%
\partial x^{1}}\in \mathcal{C}_{0}^{\infty }(\Omega ^{\sharp },C_{\sharp
}^{\infty }(\Omega ^{1}))$ in (\ref{TwoScaleModel}).

\begin{itemize}
\item \textbf{Step 1} Lemma \ref{Sixth Block} $\Longrightarrow $
\begin{equation}
\int_{\Omega ^{\sharp }\times \Omega ^{1}}a^{0}\left( \frac{\partial u^{0}}{%
\partial x^{\sharp }}+\frac{\partial \theta ^{1}}{\partial x^{1}}\frac{%
\partial u^{0}}{\partial x^{\sharp }}\right) \left( \frac{\partial v^{0}}{%
\partial x^{\sharp }}+\frac{\partial \theta ^{1}}{\partial x^{1}}\frac{%
\partial v^{0}}{\partial x^{\sharp }}\right) \text{ }dx^{\sharp
}dx^{1}=\int_{\Omega ^{\sharp }\times \Omega ^{1}}f^{0}v^{0}\text{ }%
dx^{\sharp }dx^{1}\text{.}  \label{FinalModelPreviousStep}
\end{equation}

\item \textbf{Step 2} Factorizing and definitions (\ref{coeff homogeneises})
$\Longrightarrow $%
\begin{equation*}
\int_{\Omega ^{\sharp }}a^{H}\frac{\partial u^{0}}{\partial x^{\sharp }}%
\frac{\partial v^{0}}{\partial x^{\sharp }}\text{ }dx^{\sharp }=\int_{\Omega
^{\sharp }}f^{H}v^{0}\text{ }dx^{\sharp }.
\end{equation*}
\end{itemize}
\end{proof}

\section{Rewriting strategies\label{rewriting_stratgeies:sec}}

\noindent In this section we recall the rudiments of rewriting, namely, the
definitions of terms over a signature, of substitution and of rewriting
rules. We introduce a strategy language: its syntax and semantics in terms
of partial functions. This language will allow us to express most of the
useful rewriting strategies.

\subsection{Term, substitution and rewriting rule.}

\noindent We start with an example of rewriting rule. We define a set of
rewriting variables $\mathcal{X}=\{x,y\}$ and a set of function symbols $%
\Sigma =\{f,g,a,b,c\}$. A term is a combination of elements of $\mathcal{X}%
\cup \Sigma ,$ for instance $f(x)$ or $f(a)$. The rewriting rule $%
f(x)\leadsto g(x)$ applied to a term $f(a)$ is a two-step operation. First,
it consists in matching the left term $f(x)$ with the input term $f(a)$ by
matching the two occurences of the function symbol $f,$ and by matching
the rewriting variable $x$ with the function symbol $a$. Then, the result
$g(a)$ of the rewriting operation is obtained by replacing the rewriting
variable $x$ occuring in the right hand side $g(x)$ by the subterm $a$ that
have been associated to $x$.\ In case where a substitution is possible, as
in the application of $f(b)\rightarrow g(x)$ to $f(a)$, we say that the
rewriting rule fails.

\begin{definition}
\label{terms:def} Let $\Sigma $ be a countable set of function symbols, each
symbol $f\in \Sigma $ is associated with a non-negative integer $n$, its
\emph{arity} $ar(f)$ i.e. the number of arguments of $f$. Let $\mathcal{X}$
be a countable set of variables such that $\Sigma \cap \mathcal{X}=\emptyset
$. The set of terms, denoted by $\mathcal{T}(\Sigma ,\mathcal{X})$, is
inductively defined by

\begin{itemize}
\item $\mathcal{X}\subseteq \mathcal{T}(\Sigma ,\mathcal{X})$ (i.e. every
rewriting variable is a term),

\item for all $f\in \Sigma $ of arity $n$, and all $t_{1},\ldots ,t_{n}\in
\mathcal{T}(\Sigma ,\mathcal{X})$, the expression $f(t_{1},\ldots ,t_{n})\in
\mathcal{T}(\Sigma ,\mathcal{X})$ (i.e. the application of function symbols
to terms gives rise to terms).
\end{itemize}
\end{definition}

\noindent We denote by $\Sigma _{n}$ the subset of $\Sigma $ of the function
symbols of arity $n$. For instance in the example $f$ and $g$ belong to $%
\Sigma _{1}$ while $a$ and $b$ belong to $\Sigma _{0}$. Two other common
examples of terms are the expressions $Integral(\Omega ,f(x),x)$ and \textit{%
diff}$(f(x),x)$ which represent the expressions $\int_{\Omega }f(x)\,dx$ and
$\dfrac{df(x)}{dx}$. Notice that $Integral\in \Sigma _{3},$ \textit{diff}$%
\in \Sigma _{2},$ $f\in \Sigma _{1}$ and $x,\Omega \in \Sigma _{0}$. For the
sake of simplicity we often keep the symbolic mathematical notation to
express the rewriting rules. In the following we see a term as an oriented,
ranked and rooted tree as it is usual in symbolic computation. We recall
that in a ranked tree the child order is important. For instance the tree
associated to the term $Integral(\Omega ,f(x),x)$ has $Integral$ as its root
which has three children in the order $\Omega ,$ $f,$ $x$ and $f$ has one
child $x$.



\begin{definition}
A \emph{substitution} is a function $\sigma :\mathcal{X}\rightarrow \mathcal{%
T}(\Sigma ,\mathcal{X})$ such that $\sigma (x)\neq x$ for $x\in \mathcal{X}$%
. The set of variables that $\sigma $ does not map to themselves is called
the \emph{domain} of $\sigma $, i.e. $Dom(\sigma )=\{{x\in \mathcal{X}%
\;|\;\sigma (x)\neq x}\}$. If $Dom(\sigma )=\{{x_{1},\cdots ,x_{n}}\}$ then
we might write $\sigma $ as $\sigma =\{x_{1}\mapsto t_{1},\ldots
,x_{n}\mapsto t_{n}\}$ for some terms $t_{1},..,t_{n}$. Any substitution $%
\sigma $ can be extended to a mapping $\terms\rightarrow \terms$ as follows:
for $x\in \mathcal{X},$ $\hat{\sigma}(x)=\sigma (x)$, and for any
non-variable term $s=f(s_{1},\cdots ,s_{n})$, we define $\hat{\sigma}(s)=f(%
\hat{\sigma}(s_{1}),\cdots ,\hat{\sigma}(s_{n}))$. To simplify the notation
we do not distinguish between a substitution $\sigma :\mathcal{X}\rightarrow %
\terms$ and its extension $\hat{\sigma}:\terms\rightarrow \terms$.

The \emph{application} of a substitution $\sigma$ to a term $t$, denoted by $%
\sigma(t)$, simultaneously replaces all occurrences of variables in $t$ by
their $\sigma$-images.
\end{definition}

\noindent For instance, the maping $\sigma $ defined by $\sigma (x)=a$ is a
substitution and its extension $\hat{\sigma}$ maps $f(x)$ and $g(x)$ into $%
f(a)$ and $g(a)$. 

\bigskip

\noindent A \emph{rewriting rule}, is a pair $(l,r)$ where $l$ and $r$ are
terms in $\terms$; it will also be denoted by $l\leadsto r$. We observe that
for any two terms $s,t$, there exists at most one substitution $\sigma $
such that $\sigma (s)=t$. We mention that a rewriting rule stands for the
rule application at the top position. It is more useful to be able to apply
a rule at arbitrary position, and more generally to specify the way rules
are applied. For this purpose we next present a strategy language that
allows to built strategies out of basic constructors. To this end, we
introduce strategy constructor symbols $;,\leadsto ,\oplus ,\mu ,etc$ that
do not belong to $\Sigma \cup \mathcal{X}$. Informally, the constructor
\textbf{\ }$";"$ stands for the composition, $"\oplus "$ for the left
choice, $Some$ for the application of a strategy to the immediate subterms
of the input term, $\eta (x)$ for the fail as identity constructor, $%
Child(j,s)$ applies the strategy $s$ to the $j^{\text{th}}$ immediate
subterm, $X$ is a fixed-point variable, and $\mu $ is the fixed-point or the
iterator constructor, its purpose is to define recursive strategies. For
example, the strategy $\mu X.(s;X)$ stands for $s;s;\ldots $, that is, it is
the iteration of the application of $s$ until a fixed-point is reached. The
precise semantics of these constructors is given in Definition \ref%
{semantics:strategy:def}.

\begin{definition}
\label{strategy:def} (\textbf{Strategy}) Let $\mathcal{F}$ be a finite set
of fixed-point variables. A strategy is inductively defined by the following
grammar:
\begin{equation}
s::=l\leadsto r\;\;|\;\;s;s\;\;|\;\;s\oplus s\;\;|\;\;\eta
(s)\;\;|\;\;Some(s)\;\;|\;\;Child(j,s)\;\;|\;\;X\;\;|\;\;\mu X.s
\label{strategies0:grammar}
\end{equation}%
where $j\in \mathbb{N}$ and $X \in \mathcal{F}$. The set of strategies
defined from a set of rewriting rules in $\mathcal{T}(\Sigma ,\mathcal{X})\times \mathcal{T}(\Sigma ,\mathcal{X})$ is
denoted by $\mathcal{S}_{\mathcal{T}}$.
\end{definition}

\noindent We denote by $\mathbb{F}$ the failing result of a strategy and $%
\mathcal{T}^{\ast }(\Sigma ,\mathcal{X})=\terms\cup \mathbb{F}.$

\begin{definition}
\label{semantics:strategy:def} (\textbf{Semantics of a strategy}) The
semantics of a strategy is a function $[\![.]\!]:\mathcal{S}_{\terms%
}\rightarrow (\mathcal{T}^{\ast }(\Sigma ,\mathcal{X})\rightarrow \mathcal{T}%
^{\ast }(\Sigma ,\mathcal{X}))$ defined by its application to each grammar
component:

$[\![s]\!](\fail)=\fail$

$[\![l\leadsto r]\!](t)=%
\begin{cases}
\sigma (r) & \tif\sigma (l)=t \\
\fail & \text{otherwise}%
\end{cases}%
$

$[\![s_{1};s_{2}]\!](t)=[\![s_{2}]\!]([\![s_{1}]\!](t))$

$[\![s_{1}\oplus s_{2}]\!](t)=%
\begin{cases}
\lbrack \![s_{1}]\!](t) & \tif [\![s_1]\!](t)\neq \fail \\
\lbrack \![s_{2}]\!](t) & \text{otherwise}%
\end{cases}%
$

$[\![\eta (s)]\!](t)=%
\begin{cases}
t & \tif [\![s]\!](t)=\fail \\[0pt]
\lbrack \![s]\!](t) & \text{otherwise}%
\end{cases}%
$

$[\![Some(s)]\!](t)=%
\begin{cases}
\fail & \tif ar(t)=0 \\
f(\eta (s)(t_{1}),\ldots ,\eta (s)(t_{n})) & \tif t=f(t_{1},\ldots ,t_{n})%
\tand\exists i\in \lbrack 1..n]\text{ s.t. }[\![s]\!](t_{i})\neq \fail \\
\fail & \text{ otherwise}%
\end{cases}%
$

$[\![Child(j,s)]\!](t)=%
\begin{cases}
\fail\text{ }\tif ar(t)=0,\tor t=f(t_{1},\ldots ,t_{n})\tand j>n \\
f(t_{1},\ldots ,t_{j-1},[\![s]\!](t_{j}),t_{j+1},\ldots ,t_{n})\text{ }\tif %
t=f(t_{1},\ldots ,t_{n})\tand j\leq n.%
\end{cases}%
$
\end{definition}

\noindent The semantics of the fixed-point constructor is more subtle. One
would write:
\begin{equation}
\lbrack \![\mu X.s]\!]=[\![s[X/\mu X.s]]\!]  \label{fixedpoint:def}
\end{equation}%
but this equation cannot be directly used to define $[\![\mu X.s]\!]$, since
the right-hand side contains as a subphrase the phrase whose denotation we
are trying to define. Notice that the equation (\ref{fixedpoint:def})
amounts to saying that $[\![\mu X.s]\!]$ should be the least fixed-point of
the operator $F$:
\begin{equation*}
F(X)=\lambda X^{(\mathcal{T}^{\ast }(\Sigma ,\mathcal{X})\rightarrow
\mathcal{T}^{\ast }(\Sigma ,\mathcal{X}))}\;[\![s]\!]^{(\mathcal{T}^{\ast
}(\Sigma ,\mathcal{X})\rightarrow \mathcal{T}^{\ast }(\Sigma ,\mathcal{X}))}.
\end{equation*}%
Let $D=\mathcal{T}^{\ast }(\Sigma ,\mathcal{X})\rightarrow \mathcal{T}^{\ast
}(\Sigma ,\mathcal{X})$ and define $\sqsubseteq $ a partial order on $D$ as
follows:
\begin{equation*}
w\sqsubseteq w^{\prime }\text{ iff }graph(w)\subseteq graph(w^{\prime }).
\end{equation*}%
Let $\bot $ be the function of empty graph, and let
\begin{align*}
F_{0}& =\bot \\
F_{n}& =F(F_{n-1}).
\end{align*}%
One can show, using Knaster-Tarsky fixed-point theorem \cite{Tarski55}, that
$F_{\infty }$ is the least fixed-point of the operator $F$, that is
\begin{equation*}
F(w)=w\implies F_{\infty }\sqsubseteq w.
\end{equation*}%
Such fixed point equations arises very often in giving denotational
semantics to languages with recursive features, for instance the semantics
of the loop \textquotedblleft while" of the programming languages \cite[\S %
9, \S 10]{daglib:Kenn}.

\begin{example}
\label{UsualCompositeStrategies}

Out of the basic constructors of strategies given in Definition \ref%
{strategy:def}, we built up some useful strategies. The strategy $TopDown(s)
$ applies the strategy $s$ to an input term $t$ in a top down way starting
from the root, it stops when it succeeds. That is, if the strategy $s$
succeeds on some subterm $t^{\prime }$ of $t$, then it is not applied to the
proper subterms of $t^{\prime }$. The strategy $OuterMost(s)$ behaves
exactly like $TopDown(s)$ apart that if the strategy $s$ succeeds on some
subterm $t^{\prime }$ of $t$, then it is also applied to the proper subterms
of $t^{\prime }$. The strategy $BottomUp(s)$ (resp. $InnerMost(s)$) behaves
like $BottomUp(s)$ (resp. $InnerMost(s)$) but in the opposite direction,
i.e. it traverses a term $t$ starting from the leafs. The strategy $%
Normalizer(s)$ iterates the application of $s$ until a fixed-point is
reached. The formal definition of these strategies follows:
\begin{align*}
TopDown(s) &:=\mu X.(s\oplus Some(X)), \\
OuterMost(s) &:=\mu X.(s ; Some(X)), \\
BottomUp(s) &:=\mu X. (Some(X) \oplus s), \\
InnerMost(s) &:=\mu X.(Some(X); s), \\
Normalizer(s) &:=\mu X.(s;X).
\end{align*}
\end{example}


\begin{example}
\label{Ex of strategy}Let the variable set $\mathcal{X}=\{y,z,t,w\}$ and the
partition $\Sigma =\Sigma _{0}\cup \Sigma _{1}\cup \Sigma _{2}$ of the set
of function symbols with respect to their arity with $\Sigma
_{0}=\{x,x^{1},x^{2},\partial {\Omega ,\Omega ,\varepsilon }\},$ $\Sigma
_{1}=\{{u,v,n,O,B}\},$ $\Sigma _{2}=\{$derivative$\},$ $\Sigma _{3}=\{$%
Integral$\}$ with obvious definitions. We present the strategy that rewrites
the expression%
\begin{equation*}
\Psi =\int_{\partial {\Omega }}u(x)\text{ }n(x)\text{ }B(v(x^{1},x^{2}))%
\;dx-\int_{\Omega }u(x)\;\frac{d}{dx}({B(v(x^{1},x^{2})))}\text{ }%
dx+O(\varepsilon ),
\end{equation*}%
taking into account that $B(v)$ vanishes on the boundary $\partial {\Omega }$%
. This term is written under mathematical form for simplicity, but in
practice it is written from the above defined symbol of functions. Remark
that the expression $B(v(x^{1},x^{2}))$ is a function of the variable $x$
but this does not appear explicitly in this formulation. Such a case cannot
appear when the grammar for terms introduced in the next section is used. We
need the two rewriting rules%
\begin{align*}
r_{1}& :=\int_{\partial \Omega }w\text{ }dt\leadsto \int_{\partial \Omega }w
\text{ }dt, \\
r_{2}& :=B(v(z,y))\leadsto 0,
\end{align*}%
and the strategy $TopDown$ already defined. Notice that the rule $r_{1}$ has
not effect but to detect the presence of the integral over the boundary.
Finally, the desired strategy is:%
\begin{equation*}
F:=TopDown(r_{1};TopDown(r_{2})),
\end{equation*}%
and the result is
\begin{equation*}
\lbrack \![F]\!](\Psi )=\int_{\partial {\Omega }}u(x)\text{ }n(x)\text{ }%
B(0)\;dx-\int_{\Omega }u(x)\;\frac{d}{dx}({B(v(x}^{1},x^{2}){))}\text{ }%
dx+O(\varepsilon ).
\end{equation*}
\end{example}

\subsection{Rewriting modulo equational theories}

\noindent So far the semantics of strategies does not take into account the
properties of some function symbols, e.g. associativity and commutativity
equalities of "+". In particular the application of the rule $a+b\leadsto
f(a,b)$ to the term $(a+c)+b$ fails. More generally we next consider the
rewriting modulo an equational theory, i.e. a theory that is axiomatized by
a set of equalities.

\noindent For the sake of illustration, we consider the commutativity and
associativity theory of $+,$ $E=\{x+y=y+x,(x+y)+z=x+(y+z)\}$ and the rewrite
rule $f(x+y)\leadsto f(x)+f(y)$ applying the linearity rule of a function $f$%
. Its application to the term $f((a+b)+c)$ modulo $E$ yields the set of
terms $\{{f(a+b)+f(c),}$ ${f(a)+f(b+c),}$ ${f(b)+f(a+c)}\}.$ In the
following, we define part of the semantics of a strategy modulo a theory, we
use the notation $\mathcal{P}(\terms)$ to denote the set of subsets of $%
\terms$.

\begin{definition}
(\textbf{Semantics of a strategy modulo}) Let be $E$ be a finitary
equational theory, the semantics of a strategy modulo $E$ is a function ${%
[\![.]\!]}^{E}:\mathcal{S}_{\terms}\rightarrow (\mathcal{P(T}^{\ast }(\Sigma
,\mathcal{X}))\rightarrow \mathcal{P(T}^{\ast }(\Sigma ,\mathcal{X})))$ that
is partly defined by%
\begin{align*}
& {[\![s]\!]}^{E}(\{{t_{1},\ldots ,t_{n}}\})=\cup _{i=1}^{n}{[\![s]\!]}^{E}({%
t_{i}}) \\
& {[\![l\leadsto r]\!]}^{E}({t_{1}})=\cup _{j}\{{\sigma _{j}(r)}\}\tif %
E\implies \sigma _{j}(l)=t, \\
& [\![s_{1};s_{2}]\!]^{E}(t)=[\![s_{2}]\!]^{E}([\![s_{1}]\!]^{E}(t)) \\
& [\![s_{1}\oplus s_{2}]\!]^{E}(t)=%
\begin{cases}
\lbrack \![s_{1}]\!]^{E}(t) & \tif [\![s_1]\!](t)\neq \{{\fail}\} \\
\lbrack \![s_{2}]\!]^{E}(t) & \text{otherwise}%
\end{cases}
\\
& [\![\eta (s)]\!]^{E}(t)=%
\begin{cases}
\{{t}\} & \tif [\![s]\!]^{E}(t)=\{{\fail}\} \\[0pt]
\lbrack \![s]\!]^{E}(t) & \text{otherwise.}%
\end{cases}%
\end{align*}
\end{definition}

\noindent The semantics of $Some$ and $Child$ is more complex and we do not
detail it here. The semantics of the fixed-point operator is similar to the
one given in the rewriting modulo an empty theory.





\subsection{Conditional rewriting}

\noindent Rewriting with conditional rules, also known as conditional
rewriting, extends the basic rewriting with the notion of condition. A
conditional rewrite rule is a triplet:
\begin{equation*}
(l,r,c)
\end{equation*}%
where $c$ is a constraint  expressed in some logic. The semantics of the
rule application is given by
\begin{equation*}
{[\![(l,r,c)]\!]}^{E}(t)=%
\begin{cases}
\cup _{j}\{{\sigma _{j}(r)}\} & \tif \text{ the formula } \sigma_{j}(c)
\text{ can be derived from } E, \\
\fail & \text{otherwise.}%
\end{cases}%
\end{equation*}

The set of strategies defined over rewriting rules $(l,r,c)\in \mathcal{T}%
\times \mathcal{T}\times \mathcal{T}_{c}$ is denoted by $\mathcal{S}_{%
\mathcal{T},\mathcal{T}_{c}}.$

\subsection{Rewriting with memory}

Some definitions or computations require storing the history of the
transformations of some terms. To carry on, we introduce a particular
function symbol $\mathbb{M}\in \Sigma _{2}$ of arity two to represent the
memory. Intuitively the term $\mathbb{M}(t_{1},t_{2})$ represents the term $%
t_{1}$, besides the additional information that $t_{2}$ was transformed to $%
t_{1}$ at an early stage. From this consideration if follows that any
strategy applied to $\mathbb{M}(t_{1},t_{2})$ should only be applied to $%
t_{1}$. Formally, we define the semantics of strategy application taking
into account the memory as a partial function: ${[\![.]\!]}_{_{\mathbb{M}}}:%
\mathcal{S}_{\terms}\rightarrow (\mathcal{T}^{\ast }(\Sigma ,\mathcal{X}%
)\rightarrow \mathcal{T}^{\ast }(\Sigma ,\mathcal{X}))$ so that:

$[\![s]\!]_{\mathbb{M}}(t)= \mathbb{M}([\![s]\!]_{\mathbb{M}}(t_1), t_2) $
if $t=\mathbb{M}(t_1,t_2)$, and behaves like $[\![.]\!]$, otherwise. That is,

$[\![s]\!]_{\mathbb{M}}(\fail)=\fail$

$[\![l\leadsto r]\!]_{\mathbb{M}}(t)=%
\begin{cases}
\sigma (r) & \tif\sigma (l)=t \\
\fail & \text{otherwise}%
\end{cases}%
$

$[\![s_{1};s_{2}]\!]_{\mathbb{M}}(t)=[\![s_{2}]\!]_{\mathbb{M}%
}([\![s_{1}]\!]_{\mathbb{M}}(t))$

$[\![s_{1}\oplus s_{2}]\!]_{\mathbb{M}}(t)=%
\begin{cases}
\lbrack \![s_{1}]\!]_{\mathbb{M}}(t) & \tif [\![s_1]\!]_{\mathbb{M}}(t)\neq %
\fail \\
\lbrack \![s_{2}]\!]_{\mathbb{M}}(t) & \text{otherwise}%
\end{cases}%
$

\noindent etc.



\section{A Symbolic Computation Framework for Model Derivation \label{A
framework for modelderivation}}

\noindent In this section we propose a framework for the two-scale model
proofs. As in Example \ref{Ex of strategy}, the latter are formulated as
rewriting strategies. We notice that the following framework differs from
that used in Example \ref{Ex of strategy} in that it allows for the complete
representation of the data. It does not rely on external structures such as
hash tables. To this end, we define the syntax of the mathematical
expressions by means of a grammar $\EuScript{G}$.

\subsection{A Grammar for Mathematical Expressions\label{A Grammar}}

\noindent The grammar includes four rules to built terms for mathematical
functions $\EuScript{F}$, regions $\EuScript{R}$, mathematical variables $%
\EuScript{V}$, and boundary conditions $\EuScript{C}$. It involves $\Sigma
_{Reg}$, $\Sigma _{Var},$ $\Sigma _{Fun},$ $\Sigma _{Oper},$ and $\Sigma
_{Cons}$ which are sets of names of regions, variables, functions,
operators, and constants so subsets of $\Sigma _{0}$. Empty expressions in $%
\Sigma _{Reg}$ and $\Sigma _{Fun}$ are denoted by $\bot _{\EuScript{R}}$ and
$\bot _{\EuScript{F}}$. The set of usual algebraic operations $\Sigma
_{Op}=\{+,-,\times ,/,\symbol{94}\}$ is a subset of $\Sigma _{2}$. The
elements of $\Sigma _{Type}=\{$\textit{Unknown}$,$\textit{\ Test}$,$ \textit{%
Known}$,$\textit{\ }$\bot _{Type}\}\subset \Sigma _{0},$ $\bot _{Type}$
denoting the empty expression, are to specify the nature of a function,
namely an unknown function (as $u^{\varepsilon },$ $u^{0},$ $u^{1}$ in the
proof), a test function (as $v,$ $v^{0},$ $v^{1}$) in a weak formulation or
another known function (as $a^{\varepsilon },$ $f^{\varepsilon },$ $a^{0},$ $%
f^{0}$ or $n_{\Gamma ^{1}}$).\textbf{\ }The boundary conditions satisfied by
a function are specified by the elements of $\Sigma _{BC}=\{d,n,{pd,apd,t}%
\}\subset \Sigma _{0}$ to express that it satisfies Dirichlet, Neuman,
periodic, anti-periodic or transmission conditions. The grammar also involve
the symbols of functions $\mathtt{Reg}$, $\mathtt{Fun}$, $\mathtt{IndexedFun}
$, $\mathtt{IndexedReg}$, $\mathtt{IndexedVar}$, $\mathtt{Oper}$, $\mathtt{%
Var}$, and $\mathtt{BC}$ that define regions, mathematical functions,
indexed functions or regions or variables, operators, mathematical variables
and boundary conditions. The grammar reads as%
\begin{align*}
\EuScript{F}::=& \;\circledast (\EuScript{F},\EuScript{F})\;\;|\;\;d\;\;|\;\;%
\EuScript{V}\;\;\;| \\
& \mathtt{Fun}(f,[\EuScript{V},\ldots ,\EuScript{V}],[\EuScript{C},\ldots ,%
\EuScript{C}],K)\;\;| \\
& \mathtt{IndexedFun}(\EuScript{F},\EuScript{V})\;\;| \\
& \mathtt{Oper}(A,[\EuScript{F},\ldots ,\EuScript{F}],[\EuScript{V},\ldots ,%
\EuScript{V}],[\EuScript{V},\ldots ,\EuScript{V}],[d,\ldots ,d])\;\;| \\
& \bot _{\EuScript{F}}\;\;\;|\;\;\mathbb{M}(\EuScript{F},\EuScript{F}), \\
\EuScript{R}::=\;& \mathtt{Reg}(\Omega ,[d,\ldots ,d],\{{\EuScript{R},\ldots
,\EuScript{R}}\},\EuScript{R},\EuScript{F})\;\;| \\
& \mathtt{IndexedReg}(\EuScript{F},\EuScript{V})\;\;| \\
& \bot _{\EuScript{R}}\;\;|\;\;\mathbb{M}(\EuScript{R},\EuScript{R}), \\
\EuScript{V}::=\;& \mathtt{Var}(x,\EuScript{R})\;\;|\;\;\mathtt{IndexedVar}(%
\EuScript{V},\EuScript{V})\;\;|\;\;\mathbb{M}(\EuScript{V},\EuScript{V}), \\
\EuScript{C}::=\;& \mathtt{BC}(c,\EuScript{R},\EuScript{F})\;\;|\;\;\mathbb{M%
}(\EuScript{C},\EuScript{C}),
\end{align*}%
where the symbols $\Omega ,$ $d,$ $\circledast ,$ $f,$ $K,$ $A,$ $x$ and $c$
hold for any function symbols in $\Sigma _{Reg}$, $\Sigma _{Cons}$, $\Sigma
_{Op}$, $\Sigma _{Fun}$, $\Sigma _{Type}$, $\Sigma _{Oper}$, $\Sigma _{Var},$
and $\Sigma _{BC}$. The arguments of a region term are its region name, the
list of its space directions (e.g. [1,3] for a plane in the variables $%
(x_{1},x_{3}))$, the (possibly empty) set of subregions, the boundary and
the outward unit normal. Those of a function term are its function name, the
list of the mathematical variables that range over its domain, its list of
boundary conditions, and its nature. Those for an indexed region or variable
or function term are its function or variable term and its index (which
should be discrete). For an operator term these are its name, the list of
its arguments, the list of mathematical variable terms that it depends, the
list of mathematical variable terms of its co-domain (useful e.g. for $T$
when the image cannot be deduced from the initial set), and a list of
parameters. Finally, the arguments of a boundary condition term are its
type, the boundary where it applies and an imposed function if there is one.
For example, the imposed function is set to $0$ for an homogeneous Dirichlet
condition and there is no imposed function in a periodicity condition. We
shall denote by $\mathcal{T}_{\EuScript{R}}(\Sigma ,\emptyset ),$ $\mathcal{T%
}_{\EuScript{F}}(\Sigma ,\emptyset )$, $\mathcal{T}_{\EuScript{V}}(\Sigma
,\emptyset ),$ and $\mathcal{T}_{\EuScript{C}}(\Sigma ,\emptyset )$ the set
of terms generated by the grammar starting from the non-terminal $%
\EuScript{R},$ $\EuScript{F}$, $\EuScript{V},$ and $\EuScript{C}.$ The set
of all terms generated by the grammar (i.e. starting from $\EuScript{R},$ $%
\EuScript{F}$, $\EuScript{V},$ or $\EuScript{C}$) is denoted by $\mathcal{T}%
_{\EuScript{G}}(\Sigma ,\emptyset )$. Finally, we also define the set of
terms $\mathcal{T}_{\EuScript{G}}(\Sigma ,\mathcal{X})$ where each
non-terminal $\EuScript{R},$ $\EuScript{F}$, $\EuScript{V},$ and $%
\EuScript{C}$ can be replaced by a rewriting$\ $ variable in $\mathcal{X}$.
Equivalently, it can be generated by the extension of $\EuScript{G}$
obtained by adding " $|$ \ $x$" with $x\in \mathcal{X}$ in the definition of
each non-terminal term. Or, by adding $N::=x$, with $x\in \mathcal{X}$ for
each non-terminal $N$.

\begin{example}
Throughout this paper, an underlined symbol represents a shortcut whose
name corresponds to the term name. For instance,%
\begin{gather*}
\underline{\Omega }=\mathtt{Reg}(\Omega ,[2],\emptyset ,\underline{\Gamma },%
\underline{n}),\text{ where } \underline{\Gamma }=\mathtt{Reg}%
(\Gamma,[],\emptyset ,\bot _{\EuScript{R}},\bot _{\EuScript{F}})\text{, } \\
\underline{n}=\mathtt{Fun}(n,[\underline{x}^{\prime}],[],Known), \text{ }
\underline{x}^{\prime }=\mathtt{Var}(x,\underline{\Omega }^{\prime }) \text{
and } \underline{\Omega }^{\prime }=\mathtt{Reg}(\Omega ,[2],\emptyset ,%
\underline{\Gamma },\bot _{\EuScript{F}})\text{ }
\end{gather*}%
represents a region-term a one-dimensional domain named $\Omega $, oriented
in the direction $x_{2}$, with boundary $\underline{\Gamma }$ and with
outward unit normal $\underline{n}$. The shortcut $\underline{\Gamma }$ is
also for a region term representing the boundary named $\Gamma$. As it can
be understood from this example, except names all other fields can be void
terms or empty lists.
\end{example}

\begin{example}
An unknown function $u(x)$ defined on $\underline{\Omega }$ satisfying
homogeneous Dirichlet boundary condition $u(x)=0$ on $\underline{\Gamma }$
is represented by the function-term,
\begin{equation*}
\underline{u}(\underline{x})=\mathtt{Fun}(u,[\underline{x}],\mathtt{Cond}(d,%
\underline{\Gamma },0),\mathtt{Unknown})\text{ where }\quad \underline{x}=%
\mathtt{Var}(x,\underline{\Omega }).
\end{equation*}
\end{example}

\ucomment{
\begin{remark}%

We
shall
always
consider
that
two
parameterized
mathematical
objects
are
different once
they
differ by
at
least one of
their
parameters.
For
instance variable
term
$Var(x,Reg(\Omega
,1
,\{\bot_{\Omega}\},\bot_{\Omega},\bot_{Fun} ))$
and $%
Var(x,Reg(\Omega
^{\sharp
},,1
,\{\bot_{\Omega}\},\bot_{\Omega},%
\bot_{Fun}))$
correspond
to
two
different variables.
\end{remark}}

\subsection{Short-cut Terms\label{Short-cut Terms}}

\noindent For the sake of conciseness, we introduce shortcut terms that are
constantly used in the end of the paper: $\underline{\Omega }\in \mathcal{T}%
_{\EuScript{R}}(\Sigma ,\mathcal{X})$, $\underline{x}\in \mathcal{T}_{%
\EuScript{V}}(\Sigma ,\mathcal{X})$ defined in $\underline{\Omega }$, $%
\underline{I}\in \mathcal{T}_{\EuScript{R}}(\Sigma ,\mathcal{X})$ used for
(discrete) indices, $\underline{i}\in \mathcal{T}_{\EuScript{V}}(\Sigma ,%
\mathcal{X})$ used as an index defined in $\underline{I}$, $\underline{u}\in
\mathcal{T}_{\EuScript{F}}(\Sigma ,\mathcal{X})$ or $\underline{u}(%
\underline{x})\in \mathcal{T}_{\EuScript{F}}(\Sigma ,\mathcal{X})$ to
express that it depends on the variable $\underline{x}$ and $\underline{u}_{%
\underline{i}}$ the indexed-term of the function $\underline{u}$ indexed by $%
\underline{i}$. Similar definitions can be given for the other notations
used in the proof as $\underline{\Omega }^{\sharp },$ $\underline{x}^{\sharp
},$ $\underline{\Omega }^{1},$ $\underline{x}^{1},$ $\underline{\Omega
^{\prime }},$ $\underline{x^{\prime }},$ $\underline{v}(\underline{x}%
^{\sharp },\underline{x}^{1})$ etc. The operators necessary for the proof
are the integral, the derivative, the two-scale transform $T$, its adjoint $%
T^{\ast }$, and $B$. In addition, for some extensions of the reference proof
we shall use the discrete sum.

\noindent Instead of writing operator-terms as defined in the grammar, we
prefer to use the usual mathematical expressions. The table below establishes 
the correspondance between the two formulations.%
\begin{align*}
\int \underline{u}\,d\underline{x}& \equiv \mathtt{Oper}(\mathtt{Integral},%
\underline{u},[\underline{x}],[],[]), \\
\frac{\partial \underline{u}}{\partial \underline{x}}& \equiv \mathtt{Oper}(%
\mathtt{Partial},\underline{u},[\underline{x}],[\underline{x}],[]), \\
tr(\underline{u},\underline{x})(\underline{x^{\prime }})& \equiv \mathtt{Oper%
}(\text{Restriction},\underline{u},[\underline{x}],[\underline{x^{\prime }}%
],[]), \\
T(\underline{u},\underline{x})(\underline{x}^{\sharp },\underline{x}^{1})&
\equiv \mathtt{Oper}(T,\underline{u},[\underline{x}],[\underline{x}^{\sharp
},\underline{x}^{1}],[\varepsilon ]), \\
T^{\ast }(\underline{v},[\underline{x}^{\sharp },\underline{x}^{1}])(%
\underline{x})& \equiv \mathtt{Oper}(T^{\ast },\underline{v},[\underline{x}%
^{\sharp },\underline{x}^{1}],[\underline{x}],[\varepsilon ]), \\
B(\underline{v},[\underline{x}^{\sharp },\underline{x}^{1}])(\underline{x})&
\equiv \mathtt{Oper}(B,\underline{v},[\underline{x}^{\sharp },\underline{x}%
^{1}],[\underline{x}],[\varepsilon ]), \\
\sum_{\underline{i}}\underline{u}_{\underline{i}}& \equiv \mathtt{Oper}(%
\mathtt{Sum},\underline{u}_{\underline{i}},[\underline{i}],[],[]).
\end{align*}%
The multiplication and exponentiation involving two terms $f$ and $g$ are
written $fg$ and $f^{g}$ as usual in mathematics. All these conventions have
been introduced for terms in $\mathcal{T}(\Sigma ,\emptyset )$. For terms in
$\mathcal{T}(\Sigma ,X)$ as those encoutered in rewriting rules, the
rewriting variables can replace any of the above short cut terms.

\begin{example}
\label{Marker:GreenRule}The rewriting rule associated to the Green rule (\ref%
{GreenFormulas}) reads%
\begin{equation*}
\int \frac{\partial u}{\partial \underline{x}}v\text{ }d\underline{x}%
\leadsto -\int u\frac{\partial v}{\partial \underline{x}}\text{ }d\underline{%
x}+\int tr(u)\text{ }tr(v)\;n\text{ }d\underline{x^{\prime }}.
\end{equation*}%
with the short-cuts $\underline{\Gamma }=\mathtt{Reg}(\Gamma ,d1,\emptyset
,\bot _{\EuScript{R}},\bot _{\EuScript{F}})$, $\underline{\Omega }=\mathtt{%
Reg}(\Omega ,d2,\emptyset ,\underline{\Gamma },n)$, $\underline{x}=\mathtt{%
Var}(x,\underline{\Omega })$ and $\underline{x^{\prime }}=\mathtt{Var}(x,%
\underline{\Gamma })$. The other symbols $u,$ $v$, $x$, $\Omega ,$ $\Gamma ,$
$d1,$ $d2$, $n$ are rewriting variables, and for instance%
\begin{equation*}
\frac{\partial u}{\partial x}\equiv \mathtt{Oper}(\mathtt{Partial}%
,u,x,[],[]).
\end{equation*}%
Applying this rule according to an appropriate strategy, say the top down
strategy, to a term in $\mathcal{T}(\Sigma ,\emptyset )$ like%
\begin{equation*}
\Psi =\int \frac{\partial \underline{f}(\underline{z})}{\partial \underline{z%
}}\underline{g}(\underline{z})\text{ }d\underline{z},
\end{equation*}%
for a given variable term $\underline{z}$ and function terms $\underline{f},$
$\underline{g}$. As expected, the result is%
\begin{equation*}
-\int \underline{f}\text{ }\frac{\partial \underline{g}}{\partial \underline{%
z}}\text{ }d\underline{z}+\int \underline{f}\text{ }\underline{g}\;%
\underline{n}\text{ }d\underline{z^{\prime }}
\end{equation*}%
with evident notations for $\underline{n}$ and $\underline{z^{\prime }}$.
\end{example}

\subsection{A Variable Dependency Analyzer \label{A Variable Dependency
Analyzer}}

\noindent The \textit{variable dependency analyzer }$\Theta $\textit{\ }is
related to \textit{effect systems }in computer science \textcolor{blue}{%
\cite{Marino:2009:GTS}}. It is a function from $\mathcal{T}_{\EuScript{F}%
}(\Sigma ,\emptyset )$ to the set $\mathcal{P}(\mathcal{T}_{\EuScript{V}%
}(\Sigma ,\emptyset ))$ of the parts of $\mathcal{T}_{\EuScript{V}}(\Sigma
,\emptyset )$\textbf{. }When applied to a term $t\in \mathcal{T}_{%
\EuScript{F}}(\Sigma ,\emptyset )$, it returns the set of mathematical
variables on which $t$ depends. The analyzer $\Theta $ is used in the condition
part of some rewriting rules and is inductively defined by%
\begin{align*}
& \Theta (d)=\emptyset \text{ for }d\in \Sigma _{Cons}, \\
& \Theta (\underline{x})=\{\underline{x}\}\text{ for }\underline{x}\in
\mathcal{T}_{\EuScript{V}}(\Sigma ,\emptyset ), \\
& \Theta (\circledast (\underline{u},\underline{v}))=\Theta (\underline{u}%
)\cup \Theta (\underline{v})\text{ for }\underline{u},\underline{v}\in
\mathcal{T}_{\EuScript{F}}(\Sigma ,\emptyset )\text{ and }\circledast \in
\Sigma _{Op}, \\
& \Theta (\bot _{\EuScript{F}})=\emptyset \text{,} \\
& \Theta (\underline{u}(\underline{x^{1}},..,\underline{x^{n}}))=\{%
\underline{x^{1}},..,\underline{x^{n}}\}\text{ for }\underline{u}\in
\mathcal{T}_{\EuScript{F}}(\Sigma ,\emptyset )\text{ and }\underline{x^{1}}%
,..,\underline{x^{n}}\in \mathcal{T}_{\EuScript{V}}(\Sigma ,\emptyset ), \\
& \Theta (\underline{u}_{\underline{i}})=\Theta (\underline{u})\text{ for }%
\underline{u}\in \mathcal{T}_{\EuScript{V}}(\Sigma ,\emptyset )\text{ and }%
\underline{i}\in \mathcal{T}_{\EuScript{V}}(\Sigma ,\emptyset ), \\
& \Theta ([\underline{u^{1}},\dots ,\underline{u^{n}}])=\Theta (\underline{%
u^{1}})\cup \dots \cup \Theta (\underline{u^{n}})\text{ for }\underline{u^{1}%
},\dots ,\underline{u^{n}}\in \mathcal{T}_{\EuScript{F}}(\Sigma ,\emptyset ).
\end{align*}%
The definition of $\Theta $ on the operator-terms is done case by case,
\begin{align*}
& \Theta (\int \underline{u}\,d\underline{x})=\Theta (\underline{u}%
)\setminus \Theta (\underline{x}), \\
& \Theta (\frac{\partial \underline{u}}{\partial \underline{x}})=\left\{
\begin{array}{l}
\Theta (\underline{u})\text{ if }\Theta (\underline{x})\subseteq \Theta (u),
\\
\emptyset \text{ otherwise,}%
\end{array}%
\right. \\
& \Theta (tr(\underline{u},\underline{x})(\underline{x^{\prime }}))=\Theta (%
\underline{x^{\prime }}), \\
& \Theta (T(\underline{u},\underline{x})(\underline{x}^{\sharp },\underline{x%
}^{1}))=(\Theta (\underline{u})\setminus \Theta (\underline{x}))\cup \Theta
([\underline{x}^{\sharp },\underline{x}^{1}])\text{ if }\Theta (\underline{x}%
)\cap \Theta (\underline{u})\neq \emptyset , \\
& \Theta (T^{\ast }(\underline{v},[\underline{x}^{\sharp },\underline{x}%
^{1}])(\underline{x}))=(\Theta (\underline{v})\setminus \Theta ([\underline{x%
}^{\sharp },\underline{x}^{1}]))\cup \Theta (\underline{x})\text{ if }\Theta
([\underline{x}^{\sharp },\underline{x}^{1}])\cap \Theta (\underline{v})\neq
\emptyset , \\
& \Theta (B(\underline{v},[\underline{x}^{\sharp },\underline{x}^{1}])(%
\underline{x})))=(\Theta (\underline{v})\setminus \Theta ([\underline{x}%
^{\sharp },\underline{x}^{1}]))\cup \Theta (\underline{x})\text{ if }\Theta
([\underline{x}^{\sharp },\underline{x}^{1}])\cap \Theta (\underline{v})\neq
\emptyset , \\
& \Theta (\sum_{\underline{i}}\underline{u}_{\underline{i}})=\bigcup_{%
\underline{i}}\Theta (\underline{u}_{\underline{i}}).
\end{align*}%
We observe that these definitions are not very general, but they are
sufficient for the applications of this paper. To complete the definition of
$\Theta $, it remains to define it on memory terms,%
\begin{equation*}
\Theta (\mathbb{M}(\underline{u},\underline{v}))=\Theta (\underline{u}).
\end{equation*}

\begin{example}
For%
\begin{equation*}
\Psi =\int_{\underline{\Omega }^{\sharp }}[\int_{\underline{\Omega }^{1}}T(%
\underline{u}(\underline{x}),\underline{x})(\underline{x}^{\sharp },%
\underline{x}^{1})\frac{\partial \underline{v}(\underline{x}^{\sharp },%
\underline{x}^{1})}{\partial \underline{x}^{1}}d\underline{x}^{1}]d%
\underline{x}^{\sharp }\in \mathcal{T}_{\EuScript{F}}(\Sigma ,\emptyset ),
\end{equation*}%
the set $\Theta (\Psi )$ of mathematical variables on which $\Psi $ depends
is hence inductively computed as follows: $\Theta (\underline{u}(\underline{x%
}))=\{\underline{x}\}$, $\Theta (T(\underline{u}(\underline{x}),\underline{x}%
)(\underline{x}^{\sharp },\underline{x}^{1}))=\{\underline{x}^{\sharp },%
\underline{x}^{1}\}$, $\Theta (\underline{v}(\underline{x}^{\sharp },%
\underline{x}^{1}))=\{\underline{x}^{\sharp },\underline{x}^{1}\}$, $\Theta (%
\frac{\partial \underline{v}(\underline{x}^{\sharp },\underline{x}^{1})}{%
\partial \underline{x}^{1}})=\{\underline{x}^{\sharp },\underline{x}^{1}\}$,
$\Theta (T(\underline{u}(\underline{x}),\underline{x})$ $(\underline{x}%
^{\sharp },\underline{x}^{1})$ $\frac{\partial \underline{v}(\underline{x}%
^{\sharp },\underline{x}^{1})}{\partial \underline{x}^{1}})=\{\underline{x}%
^{\sharp },\underline{x}^{1}\}$, $\Theta (\int_{\underline{\Omega }^{1}}T(%
\underline{u}(\underline{x}),\underline{x})(\underline{x}^{\sharp },%
\underline{x}^{1})$ $\frac{\partial \underline{v}(\underline{x}^{\sharp },%
\underline{x}^{1})}{\partial \underline{x}^{1}}d\underline{x}^{1})=\{%
\underline{x}^{\sharp }\}$, and $\Theta (\Psi )=\emptyset $, that is, $\Psi $
is a constant function.
\end{example}

\subsection{Formulation of the Symbolic Framework for Model Derivation\label%
{Formulation of the Symbolic Framework for Model Derivation}}

\noindent Now we are ready to define the framework for two-scale model
derivation by rewriting. To do so, the rewriting rules are restricted to
left and right terms $(l,r)\in \mathcal{T}_{\EuScript{G}}(\Sigma ,\mathcal{X}%
)\times \mathcal{T}_{\EuScript{G}}(\Sigma ,\mathcal{X})$. Their conditions $%
c $ are formulas generated by a grammar, not explicited here, combining
terms in $\mathcal{T}_{\EuScript{G}}(\Sigma ,\mathcal{X})$ with the usual
logical operators in $\Lambda =\{\vee ,\wedge ,\rceil ,\in \}$. It also
involves operations with the dependency analyzer $\Theta $. The set of terms
generated by this grammar is denoted by $\mathcal{T}_{\EuScript{L}}(\Sigma ,%
\mathcal{X},\EuScript{G},\Theta ,\Lambda ).$

\bigskip

It remains to argue that, given a strategy $s$ in $\mathcal{S}_{\mathcal{T}_{%
\EuScript{G}}(\Sigma ,\mathcal{X}),\mathcal{T}_{\EuScript{L}}(\Sigma ,%
\mathcal{X},\EuScript{G},\Theta ,\Lambda )}$, the set of terms $\mathcal{T}_{%
\EuScript{G}}(\Sigma,\emptyset)$ is closed under the application of $s$. It
is sufficient to show that for each rewriting $r$ rule in $s$, the
application of $r$ to any term $t \in \mathcal{T}_{\EuScript{G}%
}(\Sigma,\emptyset)$ at any position yields a term in $\mathcal{T}_{%
\EuScript{G}}(\Sigma,\emptyset)$. As an example, $\mathcal{T}_{\EuScript{G}%
}(\Sigma,\emptyset)$ is not closed under the application of the rule $%
x\leadsto \underline{\Omega}$, where $x$ is a variable. But it is closed
under the application of the linearity rule $\int_{z} f+g\, dx \leadsto
\int_{z} f \, dx + \int_{z} g \, dx$ at any position, where $f,g,x,z$ are
rewriting variables. The argument is, since $\int_{z} f+g\, dx \in \mathcal{T%
}_{\EuScript{F}}(\Sigma,\emptyset)$, then $f+g \in \mathcal{T}_{\EuScript{F}%
}(\Sigma,\emptyset)$, and hence $f, g \in \mathcal{T}_{\EuScript{F}%
}(\Sigma,\emptyset)$. Thus, $\int_{z} f \, dx + \int_{z} g \, dx \in
\mathcal{T}_{\EuScript{F}}(\Sigma,\emptyset)$. That is, a term in $\mathcal{T%
}_{\EuScript{F}}(\Sigma,\emptyset)$ is replaced by a another term in $%
\mathcal{T}_{\EuScript{F}}(\Sigma,\emptyset)$. A more general setting that
deals with the closure of regular languages under specific rewriting
strategies can be found in \cite{Gascon:2009}.

\noindent A model derivation is divided into several intermediary lemmas.
Each of them is intended to produce a new property that can be expressed as
one or few rewriting rules to be applied in another part of the derivation.
Since dynamical creation of rules is not allowed, a strategy is covering one
lemma only and is operating with a fixed set of rewriting rules. The
conversion of a result of a strategy to a new set of rewriting rules is done
by an elementary external operation that is not a limitation for
generalizations of proofs. The following definition summarizes the framework
of symbolic computation developed in this paper.

\begin{definition}
The components of the quintuplet $\Xi =\langle \Sigma ,\mathcal{X},E,%
\EuScript{G},\Theta \rangle $ provide a framework for symbolic computation
to derive multi-scale models. A two-scale model derivation is expressed as a
strategy $\pi \in \mathcal{S}_{\mathcal{T}_{\EuScript{G}}(\Sigma ,\mathcal{X}%
),\mathcal{T}_{\EuScript{L}}(\Sigma ,\mathcal{X},\EuScript{G},\Theta
,\Lambda )}$ for which the semantics ${[\![\pi ]\!]}^{E}$ is applicable to
an initial expression $\Psi \in \mathcal{T}(\Sigma ,\emptyset )$.
\end{definition}

\noindent In the end of this section we argue that this framework is in the
same time relatively simple, it covers the \textit{reference model}
derivation and it allows for the extensions presented in the next section.

\bigskip

\noindent The grammar of terms is designed to cover all mathematical
expressions occuring in the proof of the \textit{reference model }as well as
of their generalizations. A term produced by the grammar includes locally
all useful information. This avoids the use of external tables and
facilitates design of rewriting rules, in particular to take into account
the context of subterms to be transformed. It allows also for local
definitions, for instance a same name of variable $x$ can be used in
different parts of a same term with different meaning, which is useful for
instance in integrals.\ A limitation regarding generalizations presented in
the next section, is that the grammar must cover by anticipation all needed
features. This drawback should be fixed in another work by supporting
generalization of grammars in the same time as generalization of proofs.

\bigskip

\noindent Each step in the proof consists in replacing parts of an
expression according to a known mathematical property. This is well done,
possibly recursively, using rewriting rules together with strategies
allowing for precise localization. Some steps need simplifications and often
use the second linearity rule of a linear operator, $A(\lambda u)=\lambda Au$
when $\lambda $ is a scalar (or is independent of the variables in the
initial set of $A$). So variable dependency of each subterm should be
determined, this is precisely what $\Theta $, the \textit{variable
dependency analyzer},\textit{\ }is producing. The other simplifications do
not require the use of $\Theta $. In addition to the grammar $\EuScript{G},$
the analyzer $\Theta $ must be upgraded in view of each new extension.

\bigskip

\noindent In all symbolic computation based on the grammar $\EuScript{G}$,
it is implicitely assumed that the derivatives, the integrals and the traces
(i.e. restriction of a function to the boundary) are well defined since the
regularity of functions is not encoded.

\bigskip

\noindent Due to the algebraic nature of the mathematical proofs, this
framework has been formulated by considering these proofs as a calculus
rather than formal proofs that can be formalized and checked with a proof
assistant \cite{L:BC04,Wong96aproof}. Indeed, this is far simpler 
 and allows, from a very small set of tools,
for building significant mathematical derivation. To cover broader proofs,
the framework must be changed by extending the grammar and the variable
dependency analyzer only. Yet, the language Tom \cite{BBK_Tom07} does not
provide a complete environment for the implementation of our framework since
it does not support the transformation of rewriting rules, despite it
provides a rich strategy language and a module for the specification of the
grammar.

\section{Transformation of Strategies as Second Order Strategies \label%
{Transformation of Strategies as Second Order Strategies}}

\noindent For a given rewriting strategy representing a model proof, one
would like to transform it to obtain a derivation of more complex models.
Transforming a strategy $\pi \in \mathcal{S}_{\terms}$ is achieved by
applying strategies to the strategy $\pi $ itself. For this purpose, we
consider two levels of strategies: the first order ones $\mathcal{S}_{\terms%
} $ as defined in Definition \ref{strategy:def}, and the strategies of
second order in such a way that second order strategies can be applied to
first order ones. That is, the second order strategies are considered as
terms in a set $\mathcal{T}(\overline{\Sigma },\overline{\mathcal{X}})$ of
terms where $\overline{\Sigma }$ and $\overline{\mathcal{X}}$ remain to be
defined. Given a set of strategies $\mathcal{S}_{\terms}$ that comes with a
set of fixed-point variables $\mathcal{F}$, we pose $\overline{\Sigma }%
\supset \Sigma \cup \{{\leadsto ,;,\oplus ,Some,Child,\eta, \mu }\}\cup
\mathcal{F} $. Let $\overline{\mathcal{X}}$ be a set of second order
rewriting variables such that $\overline{\mathcal{X}}\cap (\mathcal{X}\cup
\overline{\Sigma })=\emptyset $. Notice that first order rewriting variables
and fixed-point variables are considered as constants in $\mathcal{T}(%
\overline{\Sigma },\overline{\mathcal{X}})$, i.e. function symbols in $%
\overline{\Sigma }_{0}$. Notice also that the arity of the function symbols $%
\leadsto ,{;,}\oplus ,Child,\mu $ is two, and the arity of $Some$ and $\eta$
is one. In particular, the rule $l\leadsto r$ can be viewed as the term $%
\leadsto (l,r)$ with the symbol $\leadsto $ at the root, and the strategy $%
\mu X.s$ viewed as the term $\mu (X,s)$. This allows us to define second
order strategies $\overline{\mathcal{S}}_{\mathcal{T}(\overline{\Sigma },%
\overline{\mathcal{X}})}$ by the grammar
\begin{equation}
\bar{s}::=l\bar{\leadsto}r\;\;|\;\;\bar{s}\bar{;}\bar{s}\;\;|\;\;\bar{s}\bar{%
\oplus}\bar{s}\;\;|\;\;\bar{\eta}(\bar{s})\;\;|\;\;\overline{Some}(\bar{s}%
)\;\;|\;\;\overline{Child}(j,\bar{s})\;\;|\;\;X\;\;|\;\;\bar{\mu}X.\bar{s}
\label{strategies1:grammar}
\end{equation}%
Again we assume that the symbols $\bar{\leadsto},\overline{;},\overline{%
\oplus },\ldots $ of the second order strategies do no belong to $\overline{%
\Sigma }$. The semantics of the strategies in $\overline{\mathcal{S}}_{%
\mathcal{T}(\overline{\Sigma },\overline{\mathcal{X}})}$ are similar to the
semantics of first order strategies. In addition, we assume that second
order strategies transform first order strategies, to which they are
applied, into first order strategies. Composing several second order
strategies and applying such composition to a given first order strategy $s$
provide successive transformations of $s$.

\floatstyle{plain}\restylefloat{figure}
\begin{figure}[h]
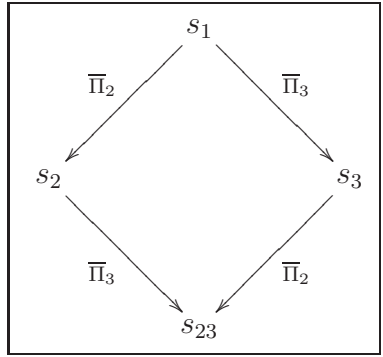

\begin{equation*}
\fbox{ \xygraph{ !{<0cm,0cm>;<2cm,0cm>:<0cm,1.5cm>} []*+{{s_1}}="S1"
[dl]*+{s_2}="S2" [rr]*+{s_3}="S3" [dl]*+{s_{23}}="S23"
"S1":"S2"_{\overline{\Pi}_2} "S1":"S3"^{\overline{\Pi}_3}
"S2":"S23"_{\overline{\Pi}_3} "S3":"S23"^{\overline{\Pi}_2} }}
\end{equation*}%
\caption{An example of the composition of transformations of strategies.}
\label{extensions:fig:1}
\end{figure}

\noindent In the following example we illustrate the extension of an
elementary strategy which is a rewriting rule.

\begin{example}
\label{extensionStrategy}For the set $\mathcal{X}=\{{i,j,x,x^{\sharp
},x^{1},u,\varepsilon }\}$ we define $s_{1},$ $s_{2},$ $s_{3},$ and $s_{23}$
four rewriting rules,
\begin{align*}
s_{1}& :=T(\frac{\partial u}{\partial x},x){(x^{\sharp },x^{1})}\leadsto
\frac{1}{\varepsilon }\frac{\partial T(u,x){(x^{\sharp },x^{1})}}{\partial {%
x^{1}}}\text{ for }x\in \Omega \text{ and }{(x^{\sharp },x^{1})}\in \Omega {%
^{\sharp }}\times \Omega {^{1},} \\
s_{2}& :=T(\frac{\partial u}{\partial x_{i}},x){(x^{\sharp },x^{1})}\leadsto
\frac{1}{\varepsilon }\frac{\partial T(u,x){(x^{\sharp },x^{1})}}{\partial {%
x_{i}^{1}}}\text{ for }x\in \Omega \text{ and }{(x^{\sharp },x^{1})}\in
\Omega {^{\sharp }}\times \Omega {^{1},} \\
s_{3}& :=T(\frac{\partial u}{\partial x},x){(x^{\sharp },x^{1})}\leadsto
\frac{1}{\varepsilon }\frac{\partial T(u,x){(x^{\sharp },x^{1})}}{\partial {%
x^{1}}}\text{ for }x\in \Omega _{j}\text{ and }{(x^{\sharp },x^{1})}\in
\Omega _{j}^{\sharp }\times \Omega _{j}^{1}{,} \\
s_{23}& :=T(\frac{\partial u}{\partial x_{i}},x){(x^{\sharp },x^{1})}%
\leadsto \frac{1}{\varepsilon }\frac{\partial T(u,x){(x^{\sharp },x^{1})}}{%
\partial {x_{i}^{1}}}\text{ for }x\in \Omega _{j}\text{ and }{(x^{\sharp
},x^{1})}\in \Omega _{j}^{\sharp }\times \Omega _{j}^{1}.
\end{align*}%
\noindent The rule $s_{1}$ is encountered in the reference proof, $s_{2}$ is
a (trivial) generalization of $s_{1}$ in the sense that it applies to
multi-dimensional regions $\Omega {^{1}}$ referenced by a set of variables $%
(x_{i}^{1})_{i}$, and $s_{3}$ is a second (trivial) generalization of $s_{1}$
on the number of sub-regions $(\Omega _{j})_{j},$ $(\Omega {_{j}^{\sharp }}%
)_{j}$ and $(\Omega _{j}^{1})_{j}$ in $\Omega $, $\Omega {^{\sharp }}$ and $%
\Omega {^{1}.}$ The rule $s_{23}$ is a generalization combining the two
previous generalizations. First, we aim at transforming the strategy $s_{1}$
into the strategy $s_{2}$ or the strategy $s_{3}$. To this end, we introduce
two second order strategies with $\overline{\mathcal{X}}=\{v,z\}$ and $%
\overline{\Sigma }\supset \{i,$ $j,$ $\Omega ,$ $\Omega ^{\sharp },$ $\Omega
^{1},$ $Partial,IndexedFun, IndexedVar, IndexedReg\},$%
\begin{align*}
\bar{\Pi}_{1}& :=\overline{OuterMost}(\frac{\partial v}{\partial z}\bar{%
\leadsto}\frac{\partial v}{\partial z_{i}}) \\
\bar{\Pi}_{2}& :=\overline{OuterMost}(\Omega \bar{\leadsto}\Omega _{j});%
\overline{OuterMost}(\Omega {^{\sharp }}\bar{\leadsto}\Omega _{j}^{\sharp });%
\overline{OuterMost}(\Omega {^{1}}\bar{\leadsto}\Omega _{j}^{1})
\end{align*}%
Notice that $\bar{\Pi}_{1}$ (resp. $\bar{\Pi}_{2}$) applies the rule $\dfrac{%
\partial v}{\partial z}\bar{\leadsto}\dfrac{\partial v}{\partial z_{i}}$
(resp. $\Omega \bar{\leadsto}\Omega _{j},$ $\Omega {^{\sharp }}\bar{\leadsto}%
\Omega _{j}^{\sharp }$, and $\Omega {^{1}}\bar{\leadsto}\Omega _{j}^{1}$) at
all of the positions \footnote{%
Notice the difference with $\overline{TopDown}$ which could not apply these
rules at any position.} of the input first order strategy so that
\begin{equation*}
\bar{\Pi}_{1}(s_{1})=s_{2}\text{ and }\bar{\Pi}_{2}(s_{1})=s_{3}.
\end{equation*}%
Once $\bar{\Pi}_{1}$ and $\bar{\Pi}_{2}$ have been defined, they can be
composed to produce $s_{23}:$%
\begin{equation*}
\bar{\Pi}_{2}\bar{\Pi}_{1}(s_{1})=s_{23}\text{ or }\bar{\Pi}_{1}\bar{\Pi}%
_{2}(s_{1})=s_{23}.
\end{equation*}%
The diagram of Figure 1 illustrates the application of $\bar{\Pi}_{1},$ $%
\bar{\Pi}_{2}$ and of their compositions.
\end{example}

The next example shows how an extension can not only change rewriting rules
but also to add new ones.

\begin{example}
To operate simplifications in the reference model, we use the strategy%
\begin{equation*}
s_{1}:=TopDown(\frac{\partial x}{\partial x}\leadsto 1).
\end{equation*}%
In the generalization to multi-dimensional regions, it is replaced by two
strategies involving the Kronecker symbol $\delta $, usually defined as $%
\delta (i,j)=1$ if $i=j$ and $\delta (i,j)=0$ otherwise,%
\begin{eqnarray*}
s_{2} &:&=TopDown\left( \frac{\partial x_{i}}{\partial y_{j}}\leadsto \delta
(i,j),\;x=y\right) , \\
s_{3} &:&=TopDown\left( \delta (i,j)\leadsto 1,\;i=j\right) , \\
s_{4} &:&=TopDown\left( \delta (i,j)\leadsto 0,\;i\neq j\right) .
\end{eqnarray*}%
The second order strategy that transforms $s_{1}$ into the strategy $%
Normalizer(s_{2}\oplus s_{3}\oplus s_{4})$ is
\begin{equation*}
\bar{\Pi}:=\overline{TopDown}(s_{1}\bar{\leadsto}s_{2}\oplus s_{3}\oplus
s_{4}).
\end{equation*}
\end{example}

\newpage

\section{Implementation and Experiments\label{Implementation and Experiments}%
}

\noindent The framework presented in Section \ref{Formulation of the
Symbolic Framework for Model Derivation} has been implemented in Maple$^{%
\registered }$. The implementation includes, the language \textit{Symbtrans}
of strategies already presented in \cite{BGL-JSC10}. The derivation of the
reference model presented in Section \ref{Skeleton:Sec}\textit{\ }has been
fully implemented. It starts from an input term which is the weak
formulation (\ref{Model1}) of the physical problem,%
\begin{equation}
\int \underline{a}\frac{\partial \underline{u}}{\partial \underline{x}}\frac{%
\partial \underline{v}}{\partial \underline{x}}\text{ }d\underline{x}=\int
\underline{f}\text{ }\underline{v}\text{ }d\underline{x},
\label{Skeleton Input}
\end{equation}%
where $\underline{a}=\mathtt{Fun}(a,[\underline{\Omega }],[$ $],Known),$ $%
\underline{u}=\mathtt{Fun}(u,[\underline{\Omega }],[\underline{Dirichlet}%
],Unknown),$ $\underline{v}=\mathtt{Fun}(u,[\underline{\Omega }],[\underline{%
Dirichlet}],Test),$ $\underline{\Omega }=\mathtt{Reg}(\Omega ,[1],\emptyset ,%
\underline{\Gamma },n_{\Omega })$, $\underline{\Gamma }=\mathtt{Reg}(\Gamma
,[$ $],\emptyset ,$ $\bot _{\EuScript{R}},$ $\bot _{\EuScript{F}}),$ $%
\underline{Dirichlet}=\mathtt{BC}(Dirichlet,\underline{\Gamma },0)$ and
where the short-cuts of the operators are those of Section \ref{Short-cut
Terms}. The information regarding the two-scale transformation is provided
through the test functions. For instance, in the first block the proof
starts with the expression%
\begin{equation*}
\Psi =\int \frac{\partial \underline{u}}{\partial \underline{x}}B(\underline{%
v}(\underline{x^{\sharp }},\underline{x^{1}})(\underline{x})\text{ }d%
\underline{x},
\end{equation*}%
where the test function $B(\underline{v}(\underline{x^{\sharp }},\underline{%
x^{1}})(\underline{x})$ is also an input, with $\underline{v}=\mathtt{Fun}%
(a,[\underline{x^{\sharp }},\underline{x^{1}}],[\underline{Dirichlet\sharp }%
],Test),$ $\underline{x^{\sharp }}=\mathtt{Var}(x^{\sharp },\underline{%
\Omega ^{\sharp }}),$ $\underline{x^{1}}=\mathtt{Var}(x^{1},\underline{%
\Omega ^{1}}),$ $\underline{\Omega }^{\sharp }=\mathtt{Reg}(\Omega ^{\sharp
},[1],\emptyset ,\underline{\Gamma ^{\sharp }},n_{\Omega ^{\sharp }}),$ $%
\underline{\Gamma ^{\sharp }}=\mathtt{Reg}(\Gamma ^{\sharp },[$ $],\emptyset
,\bot _{\EuScript{R}},\bot _{\EuScript{F}}),$ $\underline{\Omega }^{1}=%
\mathtt{Reg}(\Omega ^{1},[1],$ $\emptyset ,$ $\underline{\Gamma ^{1}}%
,n_{\Omega ^{1}}),$ $\underline{\Gamma ^{1}}=\mathtt{Reg}(\Gamma ^{1},[$ $%
],\emptyset ,\bot _{\EuScript{R}},\bot _{\EuScript{F}})$, and $\underline{%
Dirichlet\sharp }=\mathtt{BC}(Dirichlet\sharp ,\underline{\Gamma ^{\sharp }}%
,0).$

\bigskip

\noindent The proof is divided into five strategies corresponding to the
five blocks of the proof, each ending by some results transformed into
rewriting rules used in the following blocks. The rewriting rules used in
the strategies are FO-rules and can be classified into the three categories.

\begin{itemize}
\item \textit{Usual mathematical rules:} that represents the properties of
the derivation and integration operators, such as the linearity, the chain
rule, the Green rule, etc,

\item \textit{Specialized rules:} for the properties of the two-scale
calculus, as those of the two-scale transform, the approximation of $B$ by
the adjoint $T^{\ast }$ etc,

\item \textit{Auxiliary tools:} for transformations of expressions format
that are not related to operator properties such as the rule which
transforms $\psi _{1}=\psi _{2}$ into $\psi _{1}-\psi _{2}=0$.
\end{itemize}
\begin{table}[!h]
\begin{center}
\begin{tabular}{|l|c|c|c|}
\hline
& Usual Rules & Specialized Rules & Aux. Tools \\ \hline
Skeleton & 53 & 14 & 28 \\ \hline
\end{tabular}%
\end{center}
\caption{The number of first order rules used in the reference model.}
\label{stat:table1:fig}
\end{table}

\noindent The Table \ref{stat:table1:fig} summarizes the number of first
order (FO) rules, used in the reference model, by categories.

\bigskip

\noindent The reference model has been extended to cover three different
kinds of configurations. To proceed to an extension, the new model
derivation is established in a form that is as close as possible of the
\textit{reference proof. }The grammar and the dependency analyzer should be
completed. Then, the initial data is determined, and second order (SO)
strategies yielding the generalized model derivation are found and
optimized. As it has been already mentioned, $\EuScript{G}$ and $\Theta $
have already been designed to cover the three extensions.

\bigskip

\noindent The first generalization is to cover multi-dimensional regions,
i.e. $\Omega \subset \mathbb{R}^{n}$ with $n\geq 1$. When $n=2,$ the initial
term is

\begin{equation*}
\sum_{\underline{i}=1}^{n}\sum_{\underline{j}=1}^{n}\int \underline{a}_{%
\underline{i}\underline{j}}\frac{\partial \underline{u}}{\partial \underline{%
x}_{\underline{i}}}\frac{\partial \underline{v}}{\partial \underline{x}_{%
\underline{j}}}\text{ }d\underline{x}=\int \underline{f}\text{ }\underline{v}%
\text{ }d\underline{x},
\end{equation*}%
where $\underline{\Omega }=\mathtt{Reg}(\Omega ,[1,2],\emptyset ,\underline{%
\Gamma },n_{\Omega }),$ $\underline{a}_{\underline{i}\underline{j}}=\mathtt{%
Indexed}(\mathtt{Indexed}(\underline{a},\underline{j}),\underline{i}),$ $%
\underline{i}=\mathtt{Var}(i,\underline{I}),$ $\underline{I}=\mathtt{Reg}%
(I,[1,2],\emptyset ,\bot _{\EuScript{R}},\bot _{\EuScript{F}})$ and the
choice of the test function is trivially deduced. Then, the model derivation
is very similar to this of the \textit{reference model}, see \cite{LenSmi07}%
, so much so it is obtained simply by applying the SO strategy $\bar{\Pi}%
_{1} $defined in Example \ref{extensionStrategy}. This extension has been
tested on the four first blocks.

\bigskip

\noindent The second generalization transforms the \textit{reference model}
into a model with several adjacent one-dimensional regions (or intervals) $%
(\Omega _{k})_{k=1,..,m}$ so that $\Omega $ is still an interval i.e. $%
\Omega \subset \mathbb{R}$. For $m=2$, the initial term is the same as (\ref%
{Skeleton Input}) but with $\underline{\Omega }=\mathtt{Reg}(\Omega ,[1],$ $%
\{\underline{\Omega _{1}},\underline{\Omega _{2}}\},$ $\underline{\Gamma }%
,n_{\Omega })$, $\underline{\Omega _{1}}=\mathtt{Reg}(\Omega _{1},[1],$ $%
\emptyset ,$ $\underline{\Gamma _{1}},n_{\Omega _{1}}),$ and $\underline{%
\Omega _{2}}=\mathtt{Reg}(\Omega _{2},[1],$ $\emptyset ,$ $\underline{\Gamma
_{2}},n_{\Omega _{2}})$. The two-scale geometries, all variables, all kind
of functions and also the operators $B$ and $T$ are defined subregion by
subregion. All definitions and properties apply for each subregion, and the
proof steps are the same after spliting the integrals over the complete
region $\Omega $ into integrals over the subregions. The only major change
is in the fourth step where the equality $u_{1}^{0}=u_{2}^{0}$ at the
interface between $\Omega _{1}$ and $\Omega _{2}$ which is encoded as
transmission conditions in the boundary conditions of $u_{1}^{0}$ and $%
u_{2}^{0}.$

\bigskip

\noindent The third extension transforms the multi-dimensional model
obtained from the first generalization to a model related to thin
cylindrical regions, in the sense that the dimension of $\Omega $ is in the
order of $\varepsilon $ in some directions $i\in I^{\natural }$ and of the
order $1$ in the others $i\in I^{\sharp }$ e.g. $\Omega =(0,1)\times
(0,\varepsilon )$ where $I^{\natural }=\{2\}$ and $I^{\sharp }=\{1\}.$ The
boundary $\Gamma $ is split in two parts, the lateral part $\Gamma _{lat}$
and the other parts $\Gamma _{other}$ where the Dirichlet boundary
conditions are replaced by homogeneous Neuman boundary conditions i.e. $%
\frac{du^{\varepsilon }}{dx}=0$. In this special case the integrals of the
initial term are over a region which size is of the order of $\varepsilon $
so it is required to multiply each side of the equality by the factor $%
1/\varepsilon $ to work with expressions of the order of $1$. Moreover, the
macroscopic region differs from $\Omega $, it is equal to $\Omega ^{\sharp
}=(0,1)$ when the microscopic region remains unchanged. In general, the
definition of the adjoint $T^{\ast }$ is unchanged but $(Bv)(x)=v((x_{i})_{i%
\in I^{\sharp }},(x-x_{c}^{\sharp })/\varepsilon )$ where $x_{c}^{\sharp }$
is the center of the $c^{th}$ cell in $\Omega ^{\sharp }$. It follows that
the approximations (\ref{Inversion_Formula_1D}, \ref%
{First_Order_Approx_of_Bar}) are between $T^{\ast }$ and $\varepsilon B$
with $\sum_{i\in I^{\sharp }}x_{i}^{1}\frac{\partial v}{\partial
x_{i}^{\sharp }}$ instead of $\sum_{i=1}^{n}x_{i}^{1}\frac{\partial v}{%
\partial x_{i}^{\sharp }}$. With these main changes in the definitions and
the preliminary properties, the proof steps may be kept unchanged.

\begin{table}[!h]
\begin{center}
\begin{tabular}{|l|c|c|c|}
\hline
& Usual Rules & Specialized Rules & Aux. Tools \\ \hline
Multi-Dimension & 6 & 0 & 4 \\ \hline
Thin-Region & 2 & 0 & 0 \\ \hline
Multi-Region & 3 & 0 & 0 \\ \hline
\end{tabular}%
\end{center}
\caption{The number of first order rules used in the three extensions.}
\label{stat:table1-1:fig}
\end{table}

\bigskip

\noindent The mathematical formulation of the second and third extensions
has been derived. This allows for the determination of the necessary
SO-strategies, but they have not been implemented nor tested. To summarize
the results about the principle of extension of strategies, we show its
benefit through some statistics. In particular the main concerned is the
reusability and the extensibility of existing strategies. The Table 2 shows
an estimate of the number of new FO-rules for the three extensions in each
category and for the first four blocks.

\begin{table}[!h]
\begin{center}
\begin{tabular}{|l|c|c|c|}
\hline
& Usual Rules & Specialized Rules & Aux. Tools \\ \hline
Multi-Dimension & 9 & 2 & 3 \\ \hline
Thin-Region & 0 & 0 & 0 \\ \hline
Multi-Region & 1 & 0 & 0 \\ \hline
\end{tabular}%
\end{center}
\caption{The number of second order strategies used in the extension of
proofs. }
\label{stat:table2:fig}
\end{table}

\begin{table}[!h]
\par
\begin{center}
\begin{tabular}{|c|c|c|c|}
\hline
Input model & Resulting model & \% Modi. FO-rules & \% Modi. FO-strategies
\\ \hline
Reference & Multi-Dim. & 16.6\% & 5\% \\ \hline Multi-Dim. & Thin & 0 & 0 \\
\hline Thin & Multi-Reg. & 0 &  2.5\% \\ \hline
\end{tabular}%
\end{center}
\caption{The ratio of modified FO-rules and FO-strategies.}
\label{stat:table3:fig}
\end{table}

\noindent The Table 3 shows the number of SO-strategies used in each
extension. Finally, the Table 4 shows, the ratio of the modified FO-rules
and the ratio of the modified FO-strategies. The reusability ratio is high
since most of the FO-strategies defined in the skeleton model are reused.
Besides very little number of SO-strategies is used in the extensions. This
systematic way of the generation of proofs is a promising path that will be
further validated within more complex configurations for which the proofs
can not obtained by hand. In the future, we plan to introduce dedicated
tools to aid in the design of composition of several extensions.

\newpage

\newcommand{\etalchar}[1]{$^{#1}$}


\begin{thebibliography}{CKLW03}

\bibitem[ADH90]{ArbDou}
Todd Arbogast, Jim Douglas, Jr., and Ulrich Hornung.
\newblock Derivation of the double porosity model of single phase flow via
  homogenization theory.
\newblock {\em SIAM J. Math. Anal.}, 21:823--836, May 1990.

\bibitem[BB02]{BouBel}
G.~Bouchitte and M.~Bellieud.
\newblock Homogenization of a soft elastic material reinforced by fibers.
\newblock {\em Asymptotic Analysis}, 32(2):153, 2002.

\bibitem[BBK{\etalchar{+}}07]{BBK_Tom07}
Emilie Balland, Paul Brauner, Radu Kopetz, Pierre-Etienne Moreau, and Antoine
  Reilles.
\newblock {Tom: Piggybacking rewriting on Java}.
\newblock In {\em the proceedings of the 18th International Conference on
  Rewriting Techniques and Applications RTA 07}, pages 36--47, 2007.

\bibitem[BC04]{L:BC04}
Yves Bertot and Pierre Cast\'eran.
\newblock {\em Interactive Theorem Proving and Program Development. Coq'Art:
  The Calculus of Inductive Constructions}.
\newblock Texts in Theoretical Computer Science. Springer Verlag, 2004.

\bibitem[BGL]{BGL-JSC10}
W.~Belkhir, A.~Giorgetti, and M.~Lenczner.
\newblock A symbolic transformation language and its application to a
  multiscale method.
\newblock Submitted. December 2010, http://arxiv.org/abs/1101.3218v1.

\bibitem[BKKR01]{BKK+99}
{P}eter {B}orovansky, {C}laude {K}irchner, {H}{\'e}l{\`e}ne {K}irchner, and
  {C}hristophe {R}ingeissen.
\newblock {R}ewriting with strategies in {ELAN}: a functional semantics.
\newblock {\em {I}nternational {J}ournal of {F}oundations of {C}omputer
  {S}cience}, 12(1):69--95, 2001.

\bibitem[BLM96]{BouLuc}
Alain Bourgeat, Stephan Luckhaus, and Andro Mikelic.
\newblock Convergence of the homogenization process for a double-porosity model
  of immiscible two-phase flow.
\newblock {\em SIAM J. Math. Anal.}, 27:1520--1543, November 1996.

\bibitem[BLP78]{BenLio}
A.~Bensoussan, J.L. Lions, and G.~Papanicolaou.
\newblock {\em Asymptotic Methods for Periodic Structures}.
\newblock North-Holland, 1978.

\bibitem[BN98]{BaaNip}
F.~Baader and T.~Nipkow.
\newblock {\em Term rewriting and all that.}
\newblock Cambridge University Press, 1998.

\bibitem[CD99]{CioDon}
D.~Cioranescu and P.~Donato.
\newblock {\em An introduction to homogenization}.
\newblock Oxford University Press, 1999.

\bibitem[CD00]{Cas00}
J.~Casado-D{\'{\i}}az.
\newblock Two-scale convergence for nonlinear {D}irichlet problems in
  perforated domains.
\newblock {\em Proc. Roy. Soc. Edinburgh Sect. A}, 130(2):249--276, 2000.

\bibitem[CDG02]{CioDam02}
D.~Cioranescu, A.~Damlamian, and G.~Griso.
\newblock Periodic unfolding and homogenization.
\newblock {\em C. R. Math. Acad. Sci. Paris}, 335(1):99--104, 2002.

\bibitem[CDG08]{CioDam08}
D.~Cioranescu, A.~Damlamian, and G.~Griso.
\newblock The periodic unfolding method in homogenization.
\newblock {\em SIAM Journal on Mathematical Analysis}, 40(4):1585--1620, 2008.

\bibitem[CFK05]{Cirstea200551}
H.~Cirstea, G.~Faure, and C.~Kirchner.
\newblock A $\rho$-calculus of explicit constraint application.
\newblock {\em Electronic Notes in Theoretical Computer Science}, 117:51 -- 67,
  2005.
\newblock Proceedings of the Fifth International Workshop on Rewriting Logic
  and Its Applications (WRLA 2004).

\bibitem[CK01]{rhoCalIGLP-I+II-2001}
Horatiu Cirstea and Claude Kirchner.
\newblock The rewriting calculus --- {Part~I and II}.
\newblock {\em Logic Journal of the Interest Group in Pure and Applied Logics},
  9(3):427--498, May 2001.

\bibitem[CKLW03]{RewriteStrat_CHK2003}
Horatiu Cirstea, Claude Kirchner, Luigi Liquori, and Benjamin Wack.
\newblock Rewrite strategies in the rewriting calculus.
\newblock In Bernhard Gramlich and Salvador Lucas, editors, {\em {3rd
  International Workshop on Reduction Strategies in Rewriting and Programming
  }}, volume 86(4) of {\em {E}lectronic {N}otes in {T}heoretical {C}omputer
  {S}cience}, pages 18--34, Valencia, Spain, 2003. {E}lsevier.

\bibitem[GGJ09]{Gascon:2009}
Adri\`{a} Gasc\'{o}n, Guillem Godoy, and Florent Jacquemard.
\newblock Closure of tree automata languages under innermost rewriting.
\newblock {\em Electron. Notes Theor. Comput. Sci.}, 237:23--38, April 2009.

\bibitem[JZKO94]{JikKoz}
V.V. Jikov, V.~Zhikov, M.~Kozlov, and O.A. Oleinik.
\newblock {\em Homogenization of differential operators and integral
  functionals}.
\newblock Springer-Verlag, 1994.

\bibitem[Len97]{Len97}
M.~Lenczner.
\newblock Homog\'en\'eisation d'un circuit \'electrique.
\newblock {\em C. R. Acad. Sci. Paris S\'er. II b}, 324(9):537--542, 1997.

\bibitem[Len06]{Len06}
Michel Lenczner.
\newblock Homogenization of linear spatially periodic electronic circuits.
\newblock {\em NHM}, 1(3):467--494, 2006.

\bibitem[LS07]{LenSmi07}
M.~Lenczner and R.~C. Smith.
\newblock A two-scale model for an array of {AFM}'s cantilever in the static
  case.
\newblock {\em Mathematical and Computer Modelling}, 46(5-6):776--805, 2007.

\bibitem[MM09]{Marino:2009:GTS}
Daniel Marino and Todd Millstein.
\newblock A generic type-and-effect system.
\newblock In {\em Proceedings of the 4th international workshop on Types in
  language design and implementation}, TLDI '09, pages 39--50, New York, NY,
  USA, 2009. ACM.

\bibitem[SK95]{daglib:Kenn}
Kenneth Slonneger and Barry~L. Kurtz.
\newblock {\em Formal syntax and semantics of programming languages - a
  laboratory based approach}.
\newblock Addison-Wesley, 1995.

\bibitem[Tar55]{Tarski55}
Alfred Tarski.
\newblock A lattice-theoretical fixpoint theorem and its applications.
\newblock {\em The Journal of Symbolic Logic}, 5(4):370, 1955.

\bibitem[Ter03]{Terese03}
Terese.
\newblock {\em Term Rewriting Systems}, volume~55 of {\em Cambridge Tracts in
  Theor. Comp. Sci.}
\newblock Cambridge Univ. Press, 2003.

\bibitem[Won96]{Wong96aproof}
Wai Wong.
\newblock A proof checker for hol, 1996.

\end{thebibliography}
\end{document}